\documentclass[a4paper,UKenglish,cleveref, autoref, thm-restate]{lipics-v2021}

\hideLIPIcs  

\title{Listing, Verifying and Counting Lowest Common Ancestors in DAGs: Algorithms and Fine-Grained Lower Bounds} 

\titlerunning{Listing, Verifying and Counting LCAs in DAGs} 

\author{Surya Mathialagan}{Massachusetts Institute of Technology, United States}{smathi@mit.edu}{}{Supported by the MIT Akamai Presidential Fellowship.}
\author{Virginia {Vassilevska Williams}}{Massachusetts Institute of Technology, United States}{virgi@mit.edu}{}{Supported by an NSF CAREER Award, NSF Grants CCF-1528078, CCF-1514339 and CCF-1909429, a BSF Grant BSF:2012338, a Google Research Fellowship and a Sloan Research Fellowship.}
\author{Yinzhan Xu}{Massachusetts Institute of Technology, United States}{xyzhan@mit.edu}{}{Partially supported by NSF Grant CCF-1528078.}

\authorrunning{S. Mathialagan, V. Vassilevska Williams, and Y. Xu} 

\Copyright{Surya Mathialagan, Virginia Vassilevska Williams, and Yinzhan Xu} 



\category{} 

\relatedversion{} 

\nolinenumbers 

\EventEditors{John Q. Open and Joan R. Access}
\EventNoEds{2}
\EventLongTitle{42nd Conference on Very Important Topics (CVIT 2016)}
\EventShortTitle{CVIT 2016}
\EventAcronym{CVIT}
\EventYear{2016}
\EventDate{December 24--27, 2016}
\EventLocation{Little Whinging, United Kingdom}
\EventLogo{}
\SeriesVolume{42}
\ArticleNo{23}

\usepackage{colortbl,tikz}
\usetikzlibrary{shapes.geometric}

\usepackage[numbers,sort]{natbib}

\newtheorem{hypothesis}[theorem]{Hypothesis}
\newtheorem{fact}[theorem]{Fact}

\newcommand{\ZZ}{\ensuremath{\mathbb Z}}

\def \poly {\mathop{\rm poly}} 

\def\tO{\tilde{O}}

\newcommand{\Anc}{\ensuremath{\mathsf{Anc}}}
\newcommand{\lca}{\ensuremath{\mathsf{LCA}}}
\newcommand{\TOP}{\ensuremath{\pi}}
\newcommand{\eqlca}[1]{{\sf AP-Exact$#1$-LCA}}
\newcommand{\leqlca}[1]{{\sf AP-AtMost$#1$-LCA}}
\newcommand{\geqlca}[1]{{\sf AP-AtLeast$#1$-LCA}}

\newcommand{\aplca}{{\sf AP-LCA}}
\newcommand{\apklca}{{\sf AP-List-$k$-LCA}} 
\newcommand{\apkarglca}[1]{{\sf AP-List-$#1$-LCA}}
\newcommand{\alllca}{{\sf AP-All-LCA}}
\newcommand{\apvlca}{{\sf AP-Ver-LCA}}
\newcommand{\vlca}{{\sf Ver-LCA}}
\newcommand{\countlca}{{\sf AP-\#LCA}}
\newcommand{\maxsat}{{\sf Max-3-SAT}}
\newcommand{\maxksat}[1]{{\sf Max-{#1}-SAT}}

\newcommand{\fourclique}{{\sf 4-Clique}}
\newcommand{\hyperclique}[2]{{\sf $(#1, #2)$-Hyperclique}}
\newcommand{\ksat}[1]{{\sf $#1$-SAT}}
\newcommand{\maxwitness}{{\sf Max-Witness}}

\definecolor{LightCyan}{rgb}{0.88,1,1}

\def\eps{{\varepsilon}}

\begin{document}

\maketitle 

\begin{abstract}
The AP-LCA problem asks, given an $n$-node directed acyclic graph (DAG), to compute for every pair of vertices $u$ and $v$ in the DAG a lowest common ancestor (LCA) of $u$ and $v$ if one exists, i.e. a node that is an ancestor of both $u$ and $v$ but no proper descendent of it is their common ancestor. Recently [Grandoni et al. SODA'21] obtained the first sub-$n^{2.5}$ time algorithm for AP-LCA running in $O(n^{2.447})$ time. Meanwhile, the only known conditional lower bound for AP-LCA is that the problem requires $n^{\omega-o(1)}$ time where $\omega$ is the matrix multiplication exponent.

In this paper we study several interesting variants of AP-LCA, providing both algorithms and fine-grained lower bounds for them. The lower bounds we obtain are the first conditional lower bounds for LCA problems higher than $n^{\omega-o(1)}$. Some of our results include:
\begin{itemize}
    \item In any DAG, we can detect all vertex pairs that have at most two LCAs and list all of their LCAs in $O(n^\omega)$ time. This algorithm extends a result of [Kowaluk and Lingas ESA'07] which showed an $\tilde{O}(n^\omega)$ time algorithm that detects all pairs with a unique LCA in a DAG and outputs their corresponding LCAs.

    \item Listing $7$ LCAs per vertex pair in DAGs requires $n^{3-o(1)}$ time under the popular assumption that 3-uniform 5-hyperclique detection requires $n^{5-o(1)}$ time. This is surprising since essentially cubic time is sufficient to list all LCAs (if $\omega=2$).
    
    \item Counting the number of LCAs for every vertex pair in a DAG requires $n^{3-o(1)}$ time under the Strong Exponential Time Hypothesis, and $n^{\omega(1,2,1)-o(1)}$ time under the $4$-Clique hypothesis. This shows that the algorithm of [Echkardt, M\"{u}hling and Nowak ESA'07] for listing all LCAs for every pair of vertices is likely optimal.
    
    \item Given a DAG and a vertex $w_{u,v}$ for every vertex pair $u,v$, verifying whether all $w_{u,v}$ are valid LCAs requires $n^{2.5-o(1)}$ time assuming 3-uniform 4-hyperclique requires $n^{4 - o(1)}$ time. This defies the common intuition that verification is easier than computation since returning some LCA per vertex pair can be solved in $O(n^{2.447})$ time.   
\end{itemize}
\end{abstract}

\section{Introduction}

A lowest common ancestor (LCA) of two nodes $u$ and $v$ in a directed acyclic graph (DAG) is a common ancestor $c$ of $u$ and $v$ such that no proper descendent of $c$ is a common ancestor of $u$ and $v$. The AP-LCA problem asks to compute for every pair of nodes in a given DAG, some LCA, provided a common ancestor exists.

Computing LCAs is an important problem with a wide range of applications. For instance, LCA computation is a key ingredient in verification of the correctness of distributed computation (e.g. \cite{BouchitteR97}), object inheritance in object oriented programming languages such as C++ and Java (e.g. \cite{Ait-KaciBLN89,DucournauH87,HabibHS95}), and computational biology for finding the closest ancestor of species in rooted phylogenetic networks (e.g. \cite{FischerH10}). 

Computing LCAs is very well-understood in trees~\cite{Wen94,SchieberV88,Tarjan79,HarelT84,GabowBT84,ColeH05,BerkmanV94,BerkmanV93}. 
A{\"{\i}}t{-}Kaci, Boyer, Lincoln and
Nasr \cite{Ait-KaciBLN89} were one of the first to 
consider LCAs in DAGs, focusing on lattices and lower semilattices with object inheritance in mind.
Nyk{\"{a}}nen and Ukkonen~\cite{NykanenU94} obtained efficient algorithms for directed trees and asked if there is a subcubic time algorithm for AP-LCA in DAGs.

Bender, Martin Farach{-}Colton, Pemmasani, Skiena and
Sumazin~\cite{BenderFPSS05} gave the first subcubic, $O(n^{(3+\omega)/2})\leq O(n^{2.687})$, time algorithm for AP-LCA in DAGs, where $\omega<2.37286$ is the matrix multiplication exponent \cite{almanv21}. They also showed that AP-LCA is equivalent to the so-called All-Pairs Shortest LCA Distance problem. Czumaj, Kowaluk and Lingas~\cite{KowalukL05,CzumajKL07} 
improved the AP-LCA running time to $O(n^{2.575})$  using a reduction to the \maxwitness{} Product problem. With the current best bounds for rectangular matrix multiplication~\cite{GallU18}, their algorithm runs in $O(n^{2.529})$ time. 

Notice that all subcubic algorithms above would run in $\tilde{O}(n^{2.5})$ time\footnote{$\tilde{O}$ hides poly-logarithmic factors.} if $\omega=2$.
For more than a decade, this running time remained unchallenged. It seemed that AP-LCA might actually require $n^{2.5-o(1)}$ time, similar to several other $n^{2.5}$ time problems such as computing the \maxwitness{} product (see e.g. \cite{lincoln2020monochromatic}).

Recently, Grandoni, Italiano, Lukasiewicz, Parotsidis and Uznanski~\cite{grandoni2020lca} showed that this is not the case, giving an algorithm that runs in $O(n^{2.447})$ time, or in $\tilde{O}(n^{7/3})$ time if $\omega=2$.

It is not hard to show (see \cite{BenderFPSS05,CzumajKL07}) that any algorithm for AP-LCA can be used to solve Boolean Matrix Multiplication (BMM), and hence beating $O(n^\omega)$ time  for AP-LCA would likely be difficult. No higher conditional lower bounds are known for the problem. It is still open whether $O(n^\omega)$ time can actually be achieved for AP-LCA. 

Partial progresses have been made for DAGs with special structures or for variants of AP-LCA. Czumaj, Kowaluk and Lingas~\cite{CzumajKL07} showed that AP-LCA is in $O(n^\omega)$ time for low-depth DAGs. Kowaluk and Lingas~\cite{KowalukL07} showed that in ${O}(n^\omega \log n)$ time one can return an LCA for every vertex pair that has a unique LCA. 
Eckhardt, M{\"u}hling and Nowak \cite{Eckhardt2007FastLC} showed
that one can solve the \alllca{} problem, which asks to output all LCAs for every pair of vertices, in  $O(n^{\omega(1,2,1)})$ time. Here $\omega(1,2,1)\leq 3.252$ is the exponent of multiplying an $n\times n^2$ by an $n^2\times n$ matrix. AP-LCA was also studied in the weighted setting \cite{Baumgart2007weighted}, the dynamic setting \cite{Eckhardt2007FastLC} and the space-efficient setting \cite{kowaluk2008path}.

This paper considers the following questions:
\begin{enumerate}
\item Can we return all LCAs for every pair of nodes that has at most $2$ LCAs, in $\tilde{O}(n^\omega)$ time, extending Kowaluk and Lingas's algorithm~\cite{KowalukL07}?

\item The \aplca{} problem asks us to exhibit a single LCA for each vertex pair. What if we want to list 2, 3, $\ldots, k$ LCAs? How fast can we do it?

So far two variants of LCA are studied: list a single LCA per pair and list all LCAs per pair. What about listing numbers in between? This is just as natural. In phylogenetic networks for instance, there can be multiple LCAs per species pair, but typically not too many. Then listing a constant number of LCAs fast can give a better picture than listing a single representative.
Other applications of \aplca{} would similarly make more sense for listing multiple LCAs. 
\item How fast can we count the number of LCAs each vertex pair has? 

\item Suppose that for every pair of nodes $u,v$ in a DAG we are given a node $w_{u,v}$. Can we efficiently determine whether $w_{u,v}$ is an LCA of $u$ and $v$, for each $u,v$?
One would think that if \aplca{} can be solved faster than $O(n^{2.5})$ time, then this verification version of the problem should also be solvable faster.

\end{enumerate}

We provide algorithms and fine-grained conditional lower bounds to address the above questions. Our lower bounds are the first conditional lower bounds higher than $n^{\omega-o(1)}$ for LCA problems.

\subsection{Our results}

\subparagraph*{Detecting and listing $O(1)$ LCAs.}
Our results for this part are summarized in Table~\ref{table:constresults}.
\begin{table}[h]
        \centering
        \resizebox{\textwidth}{!}{
        \begin{tabular}{c|cc|cc|cc|c}
             $k$ &  \multicolumn{2}{|c|}{\geqlca{k} Exponent} &  \multicolumn{2}{|c|}{\apklca{} Exponent} & \multicolumn{2}{|c|}{Best Lower Bound} & Source of LB\\
             \hline 
             1 & $\omega$ (2) & Folklore  & ${2.447}$ (7/3) & \cite{grandoni2020lca} & $\omega$ (2) & \cite{BenderFPSS05} & BMM\\
             2 & $\omega$ (2) & \cite{KowalukL07}, Thm \ref{thm:exact_1_lca} & 2.529 (2.5) & Thm \ref{thm:2_3_listing} & $\omega$ (2) & \cite{BenderFPSS05} & BMM\\
             3 & $\omega$ (2) & Thm \ref{thm:exact_2_lca} & ${2.529}$ (2.5) & Thm \ref{thm:2_3_listing} & $\omega$ (2) & \cite{BenderFPSS05} & BMM\\
             4 & $3$ & & $3$ & & 2.5 & Thm \ref{thm:apeq3456lb} & \hyperclique{4}{3}\\
             5 & $3$ & & $3$ & & ${2.666}$ & Thm \ref{thm:apeq3456lb} & \hyperclique{5}{3}\\
             6 & $3$ & & $3$ & & ${2.8}$ & Thm \ref{thm:apeq3456lb} &\hyperclique{6}{3}\\
             \rowcolor{LightCyan}
             7 & $3$ & & $3$ & & $3$ & Thm \ref{thm:apeq3456lb} & \hyperclique{5}{3}\\
             \rowcolor{LightCyan}
             All & N/A & & $\omega(1, 2, 1)$ (3) & \cite{Eckhardt2007FastLC} & $\omega(1, 2, 1)$ (3) & Thm \ref{thm:SETH}, \ref{thm:4_clique_hardness}& SETH, \fourclique{}
        \end{tabular}}
        \caption{A summary of our results for detecting and listing LCAs. In the second and third columns, we give the best known runtime exponents for \geqlca{k} and \apkarglca{k} respectively. An exponent of 3 above corresponds to the trivial brute-force algorithm. In the fourth and fifth columns, we give the best conditional lower bounds for the exponent of \geqlca{k}, and the corresponding hardness sources for the lower bounds. The exponents and lower bounds in the last row are for \alllca{} problem. All values in parentheses are the corresponding values when $\omega = 2$. }
        \label{table:constresults}
\end{table}

Let us define \eqlca{k}, \geqlca{k} and \leqlca{k} as the problems of deciding for every pair of vertices in a given DAG, whether they have exactly, greater than or equal to, and less than or equal to $k$ LCAs, respectively.

We study how fast \eqlca{k}, \geqlca{k} and \leqlca{k} can be solved for constant $k$. More generally, we study the problem of returning $k$ LCAs per vertex pair if it has at least $k$ LCAs, or all LCAs if it has fewer. We call the latter problem \apkarglca{k}.

For any constant $k$, one can return up to $k$ LCAs for every vertex pair in a DAG in cubic time using a trivial brute-force algorithm\footnote{We first compute a topological ordering of the graph in $O(n^2)$ time and the transitive closure in $O(n^\omega)$ time using~\cite{Fischer71}. For each vertex pair $(u, v)$, we scan the vertices in the reverse order of the topological ordering, and declare the current vertex $w$  a new LCA if $w$ can reach both $u$ and $v$ and $w$ cannot reach any LCAs found so far. We stop the scan as soon as we find $k$ LCAs or reach the end of the topological ordering. Given the transitive closure, each reachability check can be finished in $O(1)$ time, so the overall running time of the algorithm is $O(kn^3)$.}. More generally, if $\omega=2$, the $O(n^{\omega(1,2,1)})$ time \alllca{} algorithm in \cite{Eckhardt2007FastLC} would also run in essentially cubic time. 

It is thus interesting to study {\em for what values of $k$, \eqlca{k}, \geqlca{k}, \leqlca{k} and \apkarglca{k} are solvable in truly subcubic, $O(n^{3-\eps})$ for $\eps>0$, time}.

We show that for every constant $k$, the listing problem \apkarglca{k} and the decision problem \geqlca{k} are subcubically equivalent. This statement appears as Theorem~\ref{thm:listing_and_detection} in the main text. Thus, the values $k$ for which one problem is in subcubic time are exactly the same for the other problem. 

We also prove a convenient equivalence between \eqlca{k}, \geqlca{k} and \leqlca{k}:

\begin{restatable}{theorem}{basicEquiv}
\label{thm:eq-leq-geq-eq}
For any constant $k \ge 0$, the running times of \eqlca{k}, \leqlca{k} and \geqlca{(k+1)} are the same up to constant factors. 
\end{restatable}

Now we can focus on \eqlca{k}, and due to the above equivalence, we also obtain results for the other variants.

Next, we extend the result of Kowaluk and Lingas~\cite{KowalukL07} for pairs with unique LCAs to pairs with two LCAs by showing that \eqlca{k} can be solved in $O(n^\omega)$ time for both $k=1,2$. Moreover, the corresponding witness LCAs can be listed in the same time.

\begin{theorem}
\label{thm:exact_1_2_lca}
    \eqlca{1} and \eqlca{2} can be solved in $O(n^\omega)$ time with high probability by Las Vegas algorithms. Moreover, finding the LCAs for vertex pairs $(u, v)$ that have exactly $1$ or $2$ LCAs can also be solved in $O(n^\omega)$ time with high probability. 
\end{theorem}
This theorem appears as Theorems~\ref{thm:exact_1_lca} and \ref{thm:exact_2_lca} in the main text. By our equivalence theorem, the same result applies to \geqlca{(k+1)} and \leqlca{k} for $k=1,2$. 

Our algorithm for \eqlca{1} is different from that of \cite{KowalukL07}. The algorithm of \cite{KowalukL07} is deterministic while ours is randomized, so it is seemingly weaker. We nevertheless include our approach to \eqlca{1} as it is simple and saves a factor of $\log n$. Additionally, our approach generalizes to \eqlca{2}.

As our techniques no longer seem to work for the case of deciding if there are exactly 3 LCAs, we turn to conditional lower bounds. We prove that under popular fine-grained hypotheses, 
the following hold in the word-RAM model with $O(\log n)$ bit words (Theorem~\ref{thm:apeq3456lb}):
\eqlca{k} requires time $n^{2.5-o(1)}$ for $k=3$, $n^{8/3-o(1)}$ for $k=4$, $n^{2.8-o(1)}$ for $k=5$ and $n^{3-o(1)}$ for $k=6$.

With our earlier equivalence theorem in mind, our conditional lower bound for \eqlca{3} means that detecting for each pair whether it has at least $4$ LCAs, or listing $4$ LCAs per vertex pair also requires $n^{2.5 - o(1)}$ time. In particular, this shows that listing 4 LCAs is more difficult than listing just one LCA per vertex pair, as the latter has an $O(n^{2.447})$ time algorithm~\cite{grandoni2020lca}.

Furthermore, our conditional lower bound for \eqlca{6}  also implies that \geqlca{7} requires $n^{3-o(1)}$ time, and hence the clearly even harder problem of listing $7$ LCAs per vertex pair requires $n^{3-o(1)}$ time. 
This is intriguing since as we mentioned earlier, we can list all LCAs per pair in essentially cubic time if $\omega=2$.

We also show the following algorithmic results for \apkarglca{2} and \apkarglca{3}. 

\begin{restatable}{theorem}{TwoThreelisting}
\label{thm:2_3_listing}
For $k = 2$ and $k = 3$, the \apklca{} problem can be deterministically solved in $\tilde{O}(n^{2 + \lambda})$ time, where $\lambda$ satisfies the equation $\omega(1, \lambda, 1) = 1 + 2 \lambda$. Here, $\omega(1, \lambda, 1)$ is the exponent of multiplying an $n\times n^\lambda$ by an $n^\lambda\times n$ matrix.
\end{restatable}

The running time for \apklca{} above matches
 the best known running time for \maxwitness{} product \cite{CzumajKL07}.
Using the current best bounds for rectangular matrix multiplication~\cite{GallU18}, the runtime we get for \apklca{} is $O(n^{2.529})$ for $k=2$ and $3$. 
\subparagraph*{Counting LCAs.}
We now turn our attention to computing the number of LCAs for every pair of vertices in a DAG. We call this problem 
\countlca{}. As shown in \cite{Eckhardt2007FastLC}, we can list all LCAs for every pair of vertices in $O(n^{\omega(1,2,1)})$ time, which is essentially cubic time if $\omega=2$. Thus in particular, we can also count all the LCAs in the same amount of time. 

One might wonder, can the counts be computed faster, in truly subcubic time? We show that under the Strong Exponential Time (SETH) Hypothesis \cite{ipz01,cip10,cip13}, this is impossible, even if we are only required to return the count for a vertex pair if it is smaller than some superconstant function $g(n)$. Notice that we can solve this restrained case in $O(n^3 g(n))$ time using the brute-force algorithm, so the following theorem is tight up to $n^{o(1)}$ factors when $g(n)$ is $\tO(1)$. 

\begin{restatable}{theorem}{SETH}
\label{thm:SETH}
Assuming SETH, \countlca{} requires $n^{3-o(1)}$ time, even if we only need to return the minimum between the count and $g(n)$ for any $g(n) = \omega(1)$. 
\end{restatable}

The current best running time $O(n^{\omega(1,2,1)})$ for listing LCAs and also for \countlca{} is actually supercubic, however. For the current best bounds on $\omega(1,2,1)$, it is $O(n^{3.252})$~\cite{GallU18}.
In fact, there are serious limitations of the known matrix multiplication techniques \cite{almanuni,ChristandlGLZ20,zuiddamccc} that show that current techniques cannot be used to prove that $\omega(1,2,1)<3.05$.

In this case, the cubic lower bound for \countlca{} under SETH would not be entirely satisfactory. We thus present a tight conditional lower bound from the \fourclique{} problem.

The \fourclique{} problem asks, does a given $n$-node graph contain a clique on $4$ nodes? The fastest known algorithm for \fourclique{} runs in $\tilde{O}(n^{\omega(1,2,1)})$ time~\cite{EisenbrandG04}, which has remained unchallenged for almost two decades. We show that an improvement over the $O(n^{\omega(1,2,1)})$ time for \countlca{} would also solve \fourclique{} faster. 

\begin{restatable}{theorem}{FourCliqueHardness}
\label{thm:4_clique_hardness}
    If the \countlca{} problem can be solved in $T(n)$ time, then \fourclique{} can be computed in $O(T(n) + n^\omega)$ time.
\end{restatable}

\subparagraph*{Verifying LCAs.} Oftentimes in algorithms, one is also concerned with the problem of verifying an answer besides computing an answer. In many cases, verification is an easier problem than computation.
For instance, even though computing the product of two $n\times n$ matrices $A$ and $B$ currently is only known to be possible in $O(n^{2.373})$ time, verifying whether the product of $A$ and $B$ is a matrix $C$ can be done in randomized $\tilde{O}(n^2)$ time. This was the basis of the Blum-Luby-Rubinfeld linearity test \cite{blumlubyrubinfeld}.

We consider the following two verification variants of \aplca{} which we call \vlca{} and \apvlca{}. In both variants, we are given an $n$-node DAG,
and for every pair of nodes $u,v$ in the DAG, we are also given a node $w_{u,v}$.
In \vlca{}, we want to determine whether all $w_{u,v}$ are LCAs for their respective pair $u,v$, i.e. that the matrix $w$ of candidate LCAs is all correct (or conversely, that there is {\em some} pair that has an incorrect entry). In the \apvlca{} variant we want to know for every $u,v$ whether $w_{u,v}$ is an LCA of $u$ and $v$, so this variant is potentially more difficult. After we compute the transitive closure of the graph, it takes $O(n)$ time to verify whether a vertex $w_{u, v}$ is indeed an LCA of $u$ and $v$. Thus, both \vlca{} and \apvlca{} can be solved in $O(n^3)$ time. No faster algorithm is known to the best of our knowledge.

Kowaluk and Lingas~\cite{KowalukL07} solved a variant of \apvlca{} concerning vertex pairs that have at most $2$ LCAs. Specifically, given one or two nodes per pair they showed how to  verify that those nodes are all the LCAs for the pair, in $O(n^\omega)$ time. However, their algorithm is not able to compute $2$ LCAs for vertex pairs that have exactly $2$ LCAs in $O(n^\omega)$ time.

Surprisingly, we provide strong evidence that \vlca{} and \apvlca{} are actually {\em harder} than \aplca{}, as \aplca{} can be solved in $O(n^{2.5-\eps})$ time for $\eps>0$, while under popular fine-grained hypotheses, \vlca{} and \apvlca{} require $n^{2.5-o(1)}$ time. 

Our first hardness result is that the running time of \apvlca{} is at least as high as that of the \maxwitness{} problem, whose current best running time is $O(n^{2.529})$ \cite{KowalukL05,GallU18}. If $\omega=2$, then \maxwitness{} would be solvable in $\tilde{O}(n^{2.5})$ time, and it is hypothesized \cite{lincoln2020monochromatic} that no $n^{2.5-o(1)}$ time algorithms exist for it. 

\begin{restatable}{theorem}{apvlbmaxwitness}
\label{thm:apvlb_maxwitness}
    If the \apvlca{} problem can be solved in $T(n)$ time, then the \maxwitness{} problem can be solved in $\tilde{O}(T(n))$ time.
\end{restatable}

Note that Czumaj, Kowaluk and Lingas's algorithm~\cite{CzumajKL07} for \aplca{} is essentially a reduction from \aplca{} to \maxwitness{}. Combined with their algorithm, the above theorem says that if \apvlca{} can be solved in $T(n)$ time, then \aplca{} can be solved in $\tO(T(n))$ time.

Our second result is the hardness of \vlca{} based on the hardness of the \hyperclique{4}{3} problem:
given  a $3$-uniform hypergraph on $n$ nodes, return whether it contains a $4$-hyperclique. This problem is hypothesized to require $n^{4-o(1)}$ time \cite{lincoln2018tight}, and solving it in $O(n^{4-\eps})$ time for $\eps>0$ would imply improved algorithms for Max-3-SAT and other problems (see \cite{lincoln2018tight} and the discussion therein).

\begin{restatable}{theorem}{apvlbhyperclique}
\label{thm:apvlb_hyperclique}
Assuming the \hyperclique{4}{3} hypothesis, \vlca{} requires $n^{2.5-o(1)}$ time. 
\end{restatable}

Thus, verifying candidate LCAs is most likely harder than finding LCAs, defying the common intuition that verification should be easier than computation.

\subsection{Paper Organization}
In Section~\ref{sec:prelim}, we give necessary definitions. In Section~\ref{sec:relationship}, we show basic relationships among \eqlca{k}, \leqlca{k} and \geqlca{k}, including Theorem~\ref{thm:eq-leq-geq-eq}. In Section~\ref{sec:eqlca}, we show $O(n^\omega)$ time algorithms for \eqlca{1} and \eqlca{2}, proving Theorem~\ref{thm:exact_1_2_lca}. In Section~\ref{sec:lcalist}, we prove the subcubic equivalence between \apklca{} and \geqlca{k} and prove Theorem~\ref{thm:2_3_listing} by giving algorithms for \apkarglca{2} and \apkarglca{3}. In Section~\ref{sec:lower_bound}, we prove several conditional lower bounds for \eqlca{k} and \countlca{}, including Theorem~\ref{thm:SETH} and Theorem~\ref{thm:4_clique_hardness}. In Section~\ref{sec:apvlca}, we show conditional lower bounds for \apvlca{}, proving Theorem~\ref{thm:apvlb_maxwitness} and Theorem~\ref{thm:apvlb_hyperclique}. Finally, in Section~\ref{sec:open}, we conclude with several open problems. 

\section{Preliminaries}
\label{sec:prelim}

\subsection{Notation}

Let $G = (V, E)$ be a DAG. For every $u, v \in V$, we use $\lca(u, v)$ to denote the set of vertices that are LCAs for vertex pair $u$ and $v$. We use $u \leadsto v$ to denote that $u$ can reach $v$ via zero or more edges and use $u \not \leadsto v$ to denote that $u$ cannot reach $v$. In particular, $u \leadsto u$ for every $u \in V$. We also use $\Anc(u)$ to denote the set of vertices that can reach $u$. For any $V' \subseteq V$, we use $G[V']$ to denote the  subgraph in $G$ induced by the vertex set $V'$.

We use $\omega < 2.37286$ to denote the matrix multiplication exponent \cite{almanv21}. For any constants $a, b, c \ge 0$, we use $\omega(a, b, c)$ to denote the exponent of multiplying an $n^a \times n^b$ matrix by an $n^b \times n^c$ matrix, in the arithmetic circuit model. 
Note that the fastest known algorithms for square \cite{almanv21} and rectangular \cite{GallU18} matrix multiplication all work in the arithmetic circuit model.

It is well-known that $\omega(a, b, c) = \omega(b, c, a)$ (see e.g. \cite{burgisser}).

\subsection{Variants of \aplca{}}

Given a DAG $G = (V, E)$, we study the following variants of \aplca{}. 

\begin{definition}[\eqlca{k}]
Decide if $|\lca(u, v)| = k$ for every pair $u, v \in V$. 
\end{definition}

\begin{definition}[\leqlca{k}]
Decide if $|\lca(u, v)| \le k$ for every pair $u, v \in V$. 
\end{definition}

\begin{definition}[\geqlca{k}]
Decide if $|\lca(u, v)| \ge k$ for every pair $u, v \in V$. 
\end{definition}

\begin{definition}[\countlca{}] Compute $|\lca(u, v)|$  for every pair $u, v \in V$.
\end{definition}

\begin{definition}[\apkarglca{k}] Compute for every pair $u, v \in V$ a list of $k$ distinct LCAs. If any pair $u, v \in V$ has fewer than $k$ LCAs, output all of their LCAs.
\end{definition}

\begin{definition}[\alllca{}]
    For every pair $u, v \in V$, output $\lca(u, v)$.
\end{definition}

\begin{definition}[\apvlca{}] 
Given a candidate vertex $w_{u, v}$ for each pair $u, v \in V$, decide if $w_{u, v} \in \lca(u, v)$ for every pair $u, v \in V$.
\end{definition}

\begin{definition}[\vlca{}] 
Given a candidate vertex $w_{u, v}$ for each pair $u, v \in V$, decide if there exists $u, v \in V$ such that $w_{u, v}$ is not an LCA for $u$ and $v$.
\end{definition}

\subsection{Fine-Grained Hypotheses}

In this section, we list the hypotheses we use in this paper. 

Eisenbrand et al. \cite{EisenbrandG04} gave the current best algorithm for \fourclique{} that runs in $O(n^{\omega(1, 2, 1)})$ time. The \fourclique{} hypothesis states that we cannot improve this algorithm much. 

\begin{hypothesis}[\fourclique{} Hypothesis \cite{ bringmann2017dichotomy, abboud2019mincut}]
    On a Word-RAM with $O(\log n)$ bit words, detecting a 4-clique in an $n$-node graph requires $n^{\omega(1, 2, 1)-o(1)}$ time.
\end{hypothesis}

\begin{hypothesis}[\hyperclique{\ell}{k} Hypothesis, \cite{lincoln2018tight}]
    Let $\ell > k > 2$ be constant integers. On a Word-RAM with $O(\log n)$ bit words, detecting whether an $n$-node $k$-uniform hypergraph contains an $\ell$-hyperclique requires $n^{\ell - o(1)}$ time.
\end{hypothesis}

Using common techniques (see e.g. \cite{williamssubcubic}), the \hyperclique{\ell}{k} hypothesis actually implies the hardness of the following unbalanced version of \hyperclique{\ell}{k}.
\begin{fact}
\label{fac:unbalanced_hyperclique}
    Assuming the \hyperclique{\ell}{k} hypothesis, on a Word-RAM with $O(\log n)$ bit words, detecting whether a $k$-uniform $\ell$-partite hypergraph with $n^{a_1}, \ldots, n^{a_\ell}$ vertices on each part for $a_1, \ldots, a_\ell > 0$ requires $n^{a_1 + \cdots + a_\ell - o(1)}$ time. 
\end{fact}

\begin{hypothesis}[\maxksat{k} Hypothesis, \cite{lincoln2018tight}]
    On a Word-RAM with $O(\log n)$ bit words, for any $k \ge 3$, given a $k$-CNF formula on $n$ variables and $\poly(n)$ clauses, determining  the maximum number of clauses that can be satisfied by a Boolean assignment of the variables requires $2^{n - o(n)}$ time.
\end{hypothesis}

\begin{hypothesis}[Strong Exponential Time Hypothesis (SETH), \cite{impagliazzo2001eth, cip10, cip13}]
    On a Word-RAM with $O(\log n)$ bit words, for every $\epsilon > 0$, there exists $k$ such that \ksat{k} on $n$ variables cannot be solved in $O(2^{(1 - \epsilon)n})$ time.
\end{hypothesis}

\begin{definition}
The \maxwitness{} product $C$ of two $n \times n$ Boolean matrices $A$ and $B$ is defined as \[C[i, j] = \max \{k \ | \ A[i, k] = B[k, j] = 1\}\]
where the maximum is defined to be $-\infty$ if no such witness exists.
\end{definition}

The best running time to compute the \maxwitness{} product is $O(n^{2+\lambda})$ where $\lambda$ satisfies the equation $\omega(1, \lambda, 1) = 1 + 2 \lambda$ \cite{CzumajKL07}. This running time is $\tO(n^{2.5})$ if $\omega = 2$. It is used as a hypothesis that this running time cannot be improved much. 

\begin{hypothesis}[\maxwitness{} Hypothesis, \cite{ lincoln2020monochromatic}]
    On a Word-RAM with $O(\log n)$ bit words, computing the \maxwitness{} product of two $n \times n$ matrices requires $n^{2.5 -o(1)}$ time.
\end{hypothesis}

\section{Relationships among \texorpdfstring{\eqlca{k}}{AP-Exactk-LCA}, \texorpdfstring{\leqlca{k}}{AP-AtMostk-LCA} and \texorpdfstring{\geqlca{k}}{AP-AtLeastk-LCA}}
\label{sec:relationship}

In this section, we consider the relationships between \eqlca{k}, \leqlca{k} and \geqlca{k}. Our results are depicted in Figure~\ref{fig:klcarelationships}.
\begin{figure}[ht]
    \centering
    \begin{tikzpicture}
        
        \node at(0, 0)  [] (eqa){\ \eqlca{0}\ };
        \node at(0, -2.5)  [] (leqa){\leqlca{0}};
        \node at(0, 2.5)  [] (geqa){\geqlca{1}};       
         
         \node at(4, 0)  [] (eqb){\ \eqlca{1}\ };
        \node at(4, -2.5)  [] (leqb){\leqlca{1}};
        \node at(4, 2.5)  [] (geqb){\geqlca{2}};       
         
         \node at(8, 0)  [] (eqc){\ \eqlca{2}\ };
        \node at(8, -2.5)  [] (leqc){\leqlca{2}};
        \node at(8, 2.5)  [] (geqc){\geqlca{3}};       
        
        \node at(11, 0)  [] (eqd){};
        \node at(11, -2.5)  [] (leqd){};
        \node at(11, 2.5)  [] (geqd){};       
        
        \node at(11.3, 0)  [] (){$\ldots$};
        \node at(11.3, -2.5)  [] (){$\ldots$};
        \node at(11.3, 2.5)  [] (){$\ldots$};       
        
        \draw[->,line width=1pt] (eqa) to[]  node[right] {} (eqb);
        \draw[->,line width=1pt] (leqa) to[]  node[right] {} (leqb);
        \draw[->,line width=1pt] (geqa) to[]  node[right] {} (geqb);
                  
        \draw[->,line width=1pt] (eqb) to[]  node[right] {} (eqc);
        \draw[->,line width=1pt] (leqb) to[]  node[right] {} (leqc);
        \draw[->,line width=1pt] (geqb) to[]  node[right] {} (geqc);
        
        \draw[->,line width=1pt] (eqc) to[]  node[right] {} (eqd);
        \draw[->,line width=1pt] (leqc) to[]  node[right] {} (leqd);
        \draw[->,line width=1pt] (geqc) to[]  node[right] {} (geqd);
        
        \draw[opacity=0.4, dashed, rounded corners=3] (geqa.north east) -- (leqa.south east) -- (leqa.south west) -- (geqa.north west) -- cycle;
        
        \draw[opacity=0.4, dashed, rounded corners=3] (geqb.north east) -- (leqb.south east) -- (leqb.south west) -- (geqb.north west) -- cycle;
        
        \draw[opacity=0.4, dashed, rounded corners=3] (geqc.north east) -- (leqc.south east) -- (leqc.south west) -- (geqc.north west) -- cycle;

	   \draw[<->,line width=1pt] (eqa) to[]  node[right] {} (leqa);
	   \draw[<->,line width=1pt] (eqa) to[]  node[right] {} (geqa);
	   
	   \draw[<->,line width=1pt] (eqb) to[]  node[right] {} (leqb);
	   \draw[<->,line width=1pt] (eqb) to[]  node[right] {} (geqb);
	   
	   \draw[<->,line width=1pt] (eqc) to[]  node[right] {} (leqc);
	   \draw[<->,line width=1pt] (eqc) to[]  node[right] {} (geqc);
    \end{tikzpicture}
\caption{Reductions between \leqlca{k}, \eqlca{k} and \geqlca{k}. All arrows in this figure represent $O(n^2)$ time reductions from an instance to another instance with the same input sizes up to constant factors. }
\label{fig:klcarelationships}
\end{figure}
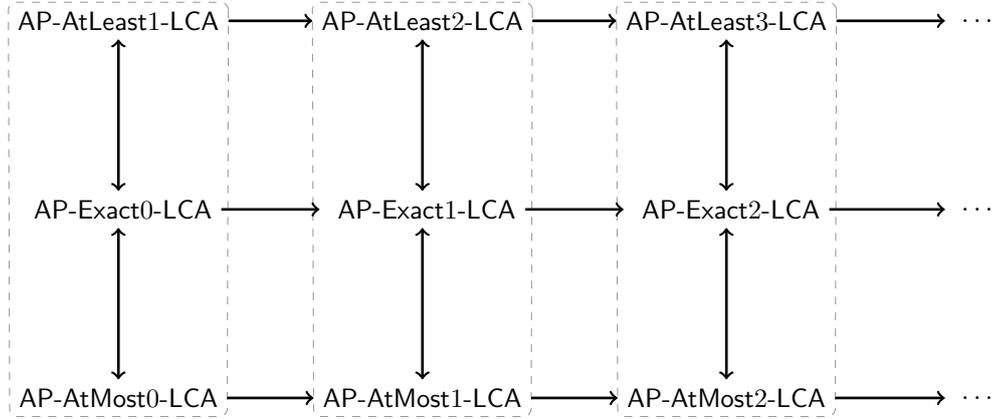

We first show the following lemma which then allows us to show that \eqlca{(k+1)} (resp. \leqlca{(k+1)}, \geqlca{(k+1)}) is harder than \eqlca{k} (resp. \leqlca{k}, \geqlca{k}) for $k \ge 0$. 

\begin{lemma}

\label{lem:one_more_LCA}
Given a DAG $G$ with $n$ vertices, we can create another DAG $G'$ with $2n+1$ vertices and a map $\rho: V(G) \rightarrow V(G')$ in $O(n^2)$ time such that for every $u, v \in V(G)$, the number of LCAs of $u$ and $v$ in $G$ is exactly one fewer than the number of LCAs of $\rho(u)$ and $\rho(v)$ in $G'$. 
\end{lemma} 

\begin{proof}
The vertex set of $G'$ contains two copies $V'$ and $V''$ of $V(G)$ and an additional vertex $x$. For every $v \in V(G)$, we use $v'$ to call its copy in $V'$ and use $v''$ to call its copy in $V''$. 
For every $(u, v) \in E(G)$, we add an edge $(u', v')$ to $G'$. For every $u \in V(G)$, we add an edge $(u', u'')$ and an edge $(x, u'')$. The function $\rho$ is defined as $\rho(u) = u''$ for every $u \in V(G)$. 

Clearly $G'$ is a DAG. It remains to show that the number of LCAs of any pair $(u, v)$ increases by $1$. 

For any pair of vertices $w$ and $v$, if $w$ can reach $v$ in $G$, then $w'$ can reach $v''$ in $G'$ since $w'$ can first take the path inside $G'[V']$ to $v'$ and then take the edge $(v', v'')$. Conversely, if $w'$ can reach $v''$ in $G'$, then $w'$ must first reach $v'$ inside $G'[V']$ so $w$ can reach $v$ in $G$ as well. This means that the set of ancestors of $v''$ in $V'$ is exactly the copy of the set of ancestors of $v$ in $G$. 

Therefore, for any pair of vertices $u, v$, the set of common ancestors of $u'', v''$ in $V'$ is exactly the copy of the set of common ancestors of $u$ and $v$. Another common ancestor of them is the vertex $x$. For any common ancestor $w_1' \in V'$, it cannot reach $x$ by the construction. Also the reachability between it and any other $w_2' \in V'$ is the same as the reachability between $w_1$ and $w_2$ in $G$ since any path from $w_1'$ to $w_2'$ must stay entirely in $G'[V']$.
Therefore, the set of LCAs of $u''$ and $v''$ in $V'$ is exactly the copy of the set of LCAs of $u$ and $v$. Additionally, $x$ is clearly an LCA. 

Thus, the number of LCAs of $u''$ and $v''$ is one more than that of $u$ and $v$. 
\end{proof}

\begin{corollary}
\label{cor:ktok+1}
For any $k \ge 0$, an instance of \eqlca{k} (resp. \leqlca{k}, \geqlca{k}) with $n$ vertices reduces to an instance of \eqlca{(k+1)} (resp. \leqlca{(k+1)}, \geqlca{(k+1)}) with  $O(n)$ vertices in $O(n^2)$ time.
\end{corollary}

\begin{proof}
Given an instance $G$ of  \eqlca{k}, we use Lemma~\ref{lem:one_more_LCA} to create $G'$ and $\rho$, and solve \eqlca{(k+1)} in $G'$. Then for any $u, v$, $|\lca(u,v)| = k$ if and only if $|\lca(\rho(u), \rho(v))| = k+1$, so we can easily use the results of \eqlca{(k+1)} on $G'$ to solve \eqlca{k} on $G$  in $O(n^2)$ time. 

The proofs for \leqlca{k} and \geqlca{k} follow similarly. 
\end{proof}
Finally, we show the relationship among \eqlca{k}, \leqlca{k} and \geqlca{k}.

\basicEquiv*

\begin{proof}
We prove a cyclic chain of reductions. 
\subparagraph*{\eqlca{k} $\rightarrow$ \leqlca{k}. }  If $k=0$, then these two are the same problem, so we assume $k \ge 1$. For any integer $\ell$, $[\ell = k] = [\ell \le k] \oplus [\ell \le (k-1)]$, where $[P]$ denotes the indicator function for a statement $P$ and $\oplus$ denotes the exclusive or operation between two Boolean values. Therefore, given an \eqlca{k} instance $G$, we can run \leqlca{k} and \leqlca{(k-1)} on $G$. Then  the number of LCAs of some vertex pair $(u, v)$ equals $k$ if and only if the two outputs for vertex pair $(u, v)$ from \leqlca{k} and \leqlca{(k-1)} are different. 

The above reduces an instance of \eqlca{k} to an instance of  \leqlca{k} and an instance of \leqlca{(k-1)}. Using Corollary~\ref{cor:ktok+1}, \leqlca{(k-1)} further reduces to \leqlca{k}. Thus, an instance of \eqlca{k} on a graph with $n$ vertices reduces to $O(1)$ instances of \leqlca{k} on graphs with $O(n)$ vertices in $O(n^2)$ time.

\subparagraph*{\leqlca{k} $\rightarrow$ \geqlca{(k+1)}. } This reduction is straightforward. For any integer $\ell$, $[\ell \le k] = \neg [\ell \ge k+1]$. Therefore, given an \leqlca{k} instance $G$, we can run \geqlca{(k+1)} and negate all the answers. 

\subparagraph*{\geqlca{(k+1)} $\rightarrow$ \eqlca{k}. } For any nonnegative integer $\ell$, $[\ell \ge k + 1] = \neg (\bigwedge_{i=0}^k [\ell = i])$. Therefore, given an \geqlca{(k+1)} instance $G$, we can run \eqlca{i} for every $0 \le i \le k$ on the same graph $G$ and use the answers to compute \geqlca{(k+1)} using the above formula. By Corollary~\ref{cor:ktok+1}, each instance of \eqlca{i} on $G$ for every $0 \le i \le k$ reduces to an instance of \eqlca{k} on a graph with $O(n)$ vertices, so an instance of \geqlca{(k+1)} on a graph with $n$ vertices reduces to $O(1)$ instances of \eqlca{k} on graphs with $O(n)$ vertices in $O(n^2)$ time.
\end{proof}

\section{Algorithms for \texorpdfstring{\eqlca{k}}{AP-Exactk-LCA}}
\label{sec:eqlca}

As noted in the introduction, \eqlca{k} can be solved in $O(n^3)$ time for any constant $k$. Interestingly,  an algorithm by Kowaluk and Lingas \cite{KowalukL07}
that finds and verifies the LCAs for vertex pairs with a unique LCA
implies that \eqlca{1} can be solved in $\tilde{O}(n^\omega)$  time deterministically. 
In this section, we present an alternative randomized algorithm for \eqlca{1}, and also extend the algorithm for \eqlca{2}.

The following claim is essential to our \eqlca{1} algorithm. 

\begin{claim}\label{claim:lca_reachability}
    Given a DAG $G = (V, E)$, for every pair of vertices $u, v \in V$, we have that 
    \begin{equation}\label{eq:lca_reachability}
        \Anc(u) \cap \Anc(v) = \bigcup_{w \in \lca(u, v)} \Anc(w).
    \end{equation}
    Moreover, if $\Anc(u) \cap \Anc(v) = \bigcup_{w \in S} \Anc(w)$ for some $S \subseteq V$, it must be the case that $\lca(u, v) \subseteq S$.
\end{claim}

\begin{proof}
        First, we will prove \eqref{eq:lca_reachability}. If a vertex $x$ lies in $\Anc(u) \cap \Anc(v)$, then $x$ is a common ancestor of $u$ and $v$. Clearly, $x$ has a descendent $w$ that is an $\lca$. 
        Hence, it is clear that $\Anc(u) \cap \Anc(v) \subseteq \cup_{w \in \lca(u, v)} \Anc(w)$. Moreover, for any vertex $x \notin \Anc(u) \cap \Anc(v)$, it cannot be an ancestor of any common ancestor of $u$ and $v$, and hence will not appear in $\cup_{w \in \lca(u, v)} \Anc(w)$. 
        
        Now we prove the second statement in the claim. Suppose $\Anc(u) \cap \Anc(v) = \bigcup_{w \in S} \Anc(w)$ for some $S \subseteq V$. First, any $w \in S$ must be a common ancestor of $u$ and $v$, since otherwise, $w\in \Anc(w)$ is not in $\Anc(u) \cap \Anc(v)$.
        Then, for any $x \in \lca(u, v)$, it must be the case that it lies in $\Anc(w)$ for some $w \in S$. Since $w$ must be a common ancestor of $u$ and $v$, this is only possible if $x = w$ (otherwise it contradicts with the condition that $x$ is an LCA). Therefore, we have $x \in S$, as desired. 
\end{proof}

\begin{theorem}\label{thm:exact_1_lca}
    There exists an $O(n^\omega)$ time Las Vegas algorithm for \eqlca{1} that succeeds  with high probability. Additionally, this algorithm can find the unique LCA for all pairs of vertices that have exactly $1$ LCA. 
\end{theorem}

\begin{proof}
        For every pair of vertices $u$ and $v$ with a unique LCA $w$, we rewrite Equation~(\ref{eq:lca_reachability}) as $\Anc(u) \cap \Anc(v) = \Anc(w).$ In fact, Claim~\ref{claim:lca_reachability} 
        gives us that this holds if and only if $w$ is a unique LCA of the pair $u$ and $v$. 
        
        Let $f: V \to \ZZ_p$ be a random function for some $p = \Theta(n^{10})$. 
        For every $S \subseteq V$, we will use $f(S)$ to denote $\sum_{x \in S} f(x)$. Then with high probability, for any $u, v, x \in V$, 
        \[\Anc(u) \cap \Anc(v) = \Anc(x) \quad \text{if and only if} \quad f(\Anc(u) \cap \Anc(v)) =f(\Anc(x)).\]
        To see this, note that for $S, S' \subseteq V$,  if $S \neq S'$, then $f(S) - f(S')$ is a sum of a nonzero number of independent uniform random variables from $\ZZ_p$. Thus if $S \neq S'$, then $\Pr\left[f(S) = f(S') \right] = \frac{1}{p}.$
        Since we are comparing $O(n^2)$ such sets of the form $f(\Anc(x))$ and $f(\Anc(u) \cap \Anc(v))$, by a union bound, the probability that two distinct sets collide is $O(n^4/p)$. 
        
        Therefore, it suffices to compute $f(\Anc(x))$ and $f(\Anc(u) \cap \Anc(v))$ for all $u, v, x \in V$. For each $x \in V(G)$, it is easy to compute $f(\Anc(x)) = \sum_{v \in \Anc(x)} f(v)$
        in $O(n)$ time. 
        To compute $F(u, v) = f(\Anc(u) \cap \Anc(v))$ for all $u, v \in V$, we construct the following matrices. Let $A$ be the transitive closure of $G$ and let $B[x, v] = f(x) \cdot A[x, v]$. Now, note that the $(u, v)$-th entry of $C = A^T B$ gives us \[C[u, v] = \sum_{x \in \Anc(u) \cap \Anc(v)} f(x) = F(u, v),\] as desired. Therefore, we can compute all $F(u, v)$ in $O(n^\omega)$ time. 
        
        Now, we sort the list $L = \{f(v) \mid v \in V(G)\}$ in $\tilde{O}(n)$ time. For each $u, v \in V$, we can find an arbitrary $w_{u, v}$ such that $F(u, v) = f(w_{u, v})$ in $\tO(1)$ time. Assuming none of the $O(n^2)$ sets we are interested in collide, which happens with probability at least $1 - O(n^4/p) = 1 - O(1/n^6)$, we find such a $w_{u, v}$ if and only if it is the unique LCA of $u, v \in V$. 
        
        To make this algorithm Las Vegas, we first notice that if our algorithm does not report a $w_{u, v}$, then $u$ and $v$ does not have a unique LCA. For the vertex pairs that our algorithm does find a $w_{u, v}$,
        we run \cite{KowalukL07}'s verification algorithm (Theorem 2 in \cite{KowalukL07}) to verify if each $w_{u, v}$ is in fact the unique LCA of $u, v$ in $O(n^\omega)$ time. If we find any errors, we can simply repeat the algorithm. 
\end{proof}

Now we show how to extend our \eqlca{1} algorithm to \eqlca{2}. 

\begin{theorem}\label{thm:exact_2_lca}
There exists an $O(n^\omega)$ time Las Vegas algorithm for \eqlca{2} that succeeds  with high probability.  Additionally, this algorithm can find the two LCAs for all pairs of vertices with exactly 2 LCAs.
\end{theorem}

\begin{proof}
        For all pair of vertices $u$ and $v$ with exactly two LCAs, say $a$ and $b$, we rewrite \eqref{eq:lca_reachability} as
        $\Anc(u) \cap \Anc(v) = \Anc(a) \cup \Anc(b).$
        Moreover, for any $u, v, a, b$ such that the above equation holds, it must be the case that either both $a$ and $b$ are the only LCAs of $u$ and $v$, or exactly one of them is the unique LCA (and the other is a common ancestor). We can detect the latter case with high probability by performing the algorithm as described in \autoref{thm:exact_1_lca}. 
        
        Let $f: V(G) \to \ZZ_p$ be a random function for some $p = \Theta(n^{10})$. By the same argument as \autoref{thm:exact_1_lca}, with high probability, for any $u, v, a, b \in V$, 
        \[\Anc(u) \cap \Anc(v) = \Anc(a) \cup \Anc(b) \quad \text{if and only if} \quad f(\Anc(u) \cap \Anc(v)) = f(\Anc(a) \cup \Anc(b)).\]
        Let $F(u, v) = f(\Anc(u) \cap \Anc(v))$ and $H(a, b) = f(\Anc(a) \cup \Anc(b))$. As we saw in \autoref{thm:exact_1_lca}, we can compute $F(u, v)$ in $O(n^\omega)$ time. 
        
        To compute $H(a, b)$, note that $\Anc(a) \cup \Anc(b) = \overline{\overline{\Anc(a)} \cap \overline{\Anc(b)}}$. First, we compute the transitive closure $A$ of $G$ in $O(n^\omega)$ time. Then, we construct an $n \times n$ matrix $M$ by setting $M[x, a] = 1 - A[x, a]$. 
        Now, construct another matrix $N$ by setting $N[x, b] = f(x) \cdot M[x, b]$. Then, it is easy to see that \[(M^T N) [a, b] = \sum_{x \in \overline{\Anc(a)} \cap \overline{\Anc(b)}} f(x) = f(\overline{\Anc(a)} \cap \overline{\Anc(b)}).\]
        Therefore, one can compute \[H(a, b) = f(\Anc(a) \cup \Anc(b)) = f(V) - f(\overline{\Anc(a)} \cap \overline{\Anc(b)}) = f(V) - (M^T N)[a, b]\]
        for all $a, b \in V$ in $O(n^\omega)$ time. 
    
        Now, sort $L = \{H(a, b) \mid a, b \in V\}$. For each $u, v$ which does not have a unique LCA, search for an arbitrary pair $a_{u,v}, b_{u,v}$  (if one exists) such that $F(u, v) = H(a_{u, v}, b_{u, v})$ in $\tO(1)$ time.  With probability $1 - O(1/n^6)$, we find such a pair for each $u, v$ if and only if $a_{u,v}$ and $b_{u,v}$ are the only two LCAs of $u$ and $v$.  
    
        To make this algorithm Las Vegas, we first notice that if our algorithm does not report a pair $a_{u, v}, b_{u, v}$, then $u$ and $v$ does not have exactly two LCAs. For vertex pairs for which our algorithm does find two LCA candidates, we run \cite{KowalukL07}'s verification algorithm (it is described in a remark in \cite{KowalukL07}) to verify that $a_{u,v}$ and $b_{u, v}$ are the only two LCAs of $u$ and $v$ in $O(n^\omega)$ time. If we find any errors, we can simply repeat the algorithm from the beginning.
\end{proof}

Note that our technique for \eqlca{1} and \eqlca{2} does not extend to \eqlca{3} because it would require us to list $f(\Anc(x) \cup \Anc(y) \cup \Anc(z))$ for all $x, y, z \in V$, which easily exceeds $n^\omega$ time. In fact, in Section~\ref{sec:lower_bound}, we show it is unlikely to obtain an $\tO(n^\omega)$ time algorithm for \eqlca{3} by proving that any $O(n^{2.5-\epsilon})$ time algorithm for $\epsilon > 0$ for \eqlca{3} would refute the \hyperclique{4}{3} hypothesis. Thus, \eqlca{3} is indeed (conditionally) harder than \eqlca{1} and \eqlca{2}. 

\section{\aplca{} Listing Algorithms}
\label{sec:lcalist}
In this section, we consider the \apklca{} problem. First, we show that \geqlca{k} and \apkarglca{k} are \emph{subcubically equivalent}, i.e. either both or neither have a truly subcubic time algorithm.

\begin{theorem}\label{thm:listing_and_detection}
    Suppose \geqlca{k} can be computed in $T(n)$ time for a constant $k$. Then, \apkarglca{k} can be computed in $O(\sqrt{n^3 \cdot T(n)})$ time. In particular, \geqlca{k} and \apkarglca{k} are subcubically equivalent.
\end{theorem}

\begin{proof}
    Suppose we are given a DAG $G = (V, E)$. First compute a topological ordering $\TOP{}$ of the vertices in $O(n^2)$ time, and the transitive closure $D$ in $O(n^\omega)$ time. Now, for every pair of vertices $u$ and $v$, we inductively find their $k$ topologically latest (with respect to $\TOP{}$) LCAs. 
    
    Suppose we have found the set $S(u, v)$ 
    of the topologically latest $\ell-1$ LCAs for every pair of vertices $u, v$, for some $1 \leq \ell \leq k$ with respect to  $\TOP{}$. 
    Now, partition the vertices into sets $V = V_1 \sqcup V_2 \sqcup \dots \sqcup V_{n/L}$, where $V_1$ contains the first $L$ vertices in the topological ordering, $V_2$ contains the next $L$ and so on for a parameter $L$ that we will set later.
    
    Let $\lca_{G[W]}(u, v)$ denote the set of LCAs of $u$ and $v$ in the subgraph induced by $W$ (note the distinction between $\lca_{G[W]}(u, v)$ and $\lca(u, v) \cap W$). Consider the vertex set $U_i = V_i \sqcup V_{i+1} \sqcup \dots \sqcup V_{n/L}$ and the induced subgraph $G_i = G[U_i]$. 
    \begin{claim}\label{claim:suffix_lca}
    For $u, v \in U_i$, it must be the case that $\lca_{G_i}(u, v) = \lca(u, v) \cap U_i.$
    \end{claim}
    \begin{proof}
    First notice that for any $w, w' \in U_i$, $w$ can reach $w'$ in $G$ if and only if $w$ can reach $w'$ in $G_i$ because any path from $w$ to $w'$ in $G$ does not use any vertex in $V \setminus U_i$ due to the fact that $U_i$ is a suffix of the topological ordering.
    Thus we can use the reachability between them  unambiguously. In particular, for any $w \in U_i$, the set of its descendants (resp. ancestors) in $G_i$ is the intersection between $U_i$ and  the set of its descendants (resp. ancestors) in $G$.
    For any $w \in \lca(u, v) \cap U_i$, it is clear that $w$ does not have any descendants in $U_i$ that is a common ancestor of $u$ and $v$,  since otherwise $w \notin \lca(u, v)$. Therefore, $w \in \lca_{G_i}(u, v)$. On the other hand, suppose $w \in \lca_{G_i}(u, v)$. Since $w$ is clearly in $U_i$, it suffices to show that $w \in \lca(u, v)$. Note that $w$ does not have any descendants in $U_i$ which is a common ancestor of $u$ and $v$. Moreover, all vertices in $V \setminus U_i$ occur earlier than $w$ in the topological ordering. Therefore, $w$ has no descendants in $V \setminus U_i$. Therefore, there is no $y \ne w$ such that $w \leadsto y \leadsto u$, and $w \leadsto y \leadsto v$, and thus, $w \in \lca(u, v)$, as desired.
    \end{proof}

    Now we describe our algorithm. For $i = n/L, n/L - 1, \dots, 1$,
    run \geqlca{\ell} on each $G_i$. For each $(u, v) \in V \times V$, keep track of the largest index $i_{u,v}$ where \geqlca{\ell} outputs 1, i.e. largest index such that $|\lca_{G_i} (u, v)| \geq \ell$. By Claim~\ref{claim:suffix_lca}, this must mean that $|\lca(u, v) \cap U_{i_{u,v}}| \geq \ell$ whereas $|\lca(u, v) \cap U_{i_{u,v} - 1}| < \ell$. In other words, the $\ell$th LCA lies in $V_{i_{u,v}}$.
    By Corollary~\ref{cor:ktok+1}, we can compute \geqlca{\ell} in time $O(T(n))$. Therefore, this step takes $O\left(\frac{n}{L} \cdot T(n)\right)$ time in total.
    
    Next, for each $u, v \in V$, note that the topologically $\ell$th LCA must lie in the vertex partition $V_{i_{u,v}}$, if $i_{u, v}$ exists. Therefore, it suffices to find the latest vertex $x \in V_{i_{u,v}}$ such that $x \in \Anc(u) \cap \Anc(v)$ and no $y \in S(u, v)$ is a descendent of $x$. Such an $x$ must in fact be the $\ell$th LCA. Note that these checks can be done in $O(1)$ time for each $x \in V_{i_{u,v}}$ using the transitive closure $D$. If there is no $i_{u,v}$ such that \geqlca{\ell} outputs 1, then $u$ and $v$ have fewer than $\ell$ LCAs. This step takes $O(\ell \cdot L) = O(L)$ time for each pair $u, v \in V$.
    
    Since we have to iteratively find up to $k$ LCAs per vertex pair, the overall runtime of the algorithm is  $O(n^\omega + k (\frac{n}{L} \cdot T(n) + n^2 \cdot L))$.
    Choosing $L = \sqrt{T(n)/n}$, we have a runtime of $O(\sqrt{n^3 \cdot T(n)})$. 
    
    Moreover, it is clear that if there is a subcubic algorithm for \apklca{}, we can use the same algorithm to solve \geqlca{k} with an $\tO(n^2)$ additional cost. Therefore the two problems are in fact subcubically equivalent.
\end{proof}

In Theorem~\ref{thm:exact_1_lca} and Theorem~\ref{thm:exact_2_lca}, we showed $O(n^\omega)$ time algorithms for \eqlca{1} and \eqlca{2}. By their equivalences with \geqlca{2} and \geqlca{3} respectively, we can also solve \geqlca{2} and \geqlca{3} in $O(n^\omega)$ time. By Theorem~\ref{thm:listing_and_detection}, these imply $O(n^{(\omega+3)/2})$ time algorithms for \apkarglca{2} and \apkarglca{3}. 

In the following theorem, we show that we can further improve the $O(n^{(3+\omega)/2})$ running time for \apkarglca{2} and \apkarglca{3} to  $\tilde{O}(n^{2 + \lambda})$ time where $\omega(1, \lambda, 1) = 1 + 2 \lambda$. Interestingly this running time matches the current best running time of the \maxwitness{} problem \cite{CzumajKL07}. 
For these algorithms, we use an idea from \cite{KowalukL07} about comparing the sizes of two sets for verifying whether a set of one or two vertices are all the LCAs. 

\TwoThreelisting*

\begin{proof}
We first describe the algorithm for $k = 2$, and we later comment on how to modify this algorithm for $k = 3$.
Suppose we are given a DAG $G = (V, E)$. First, compute a topological ordering $\TOP{}$ of the vertices and the transitive closure $D$ of the $G$. Now, use Czumaj et al.'s algorithm~\cite{CzumajKL07} to find the topologically latest LCA for every pair of vertices $u, v \in V$ with respect to $\TOP{}$. Denote this vertex by $\ell_1(u, v)$. We then find the topologically second latest LCA, which we denote by $\ell_2(u, v)$.  

Now, partition the vertices into sets $V = V_1 \sqcup V_2 \sqcup \dots \sqcup V_{n/L},$ where $V_1$ contains the first $L$ vertices in the ordering $T$, $V_2$, contains the next $L$ and so on. For each $(u, v) \in V \times V$, we will find the largest index $i_{u,v}$ such that $V_{i_{u,v}}$ contains $\ell_2(u, v)$.

For $i > i_{u, v},$ we must have 
\begin{equation}\label{eq:lca_reachability_2_listing}
    |\Anc(u) \cap \Anc(v) \cap V_i| = |\Anc(\ell_1(u, v)) \cap V_i|.
\end{equation}
It is clear that $\Anc(\ell_1(u, v)) \cap V_i \subseteq \Anc(u) \cap \Anc(v) \cap V_i$. Therefore, it suffices to show that when $i > i_{u, v}$, the reverse inclusion holds. Fix any $w \in \Anc(u) \cap \Anc(v) \cap V_i$. Since $w$ is a common ancestor of $u$ and $v$, it must be able to reach some vertex in $\lca(u, v)$. However, since $i>i_{u, v}$, the only such LCA $w$ can reach is $\ell_1(u, v)$, so $w \in \Anc(\ell_1(u, v))$. Therefore, $w \in \Anc(\ell_1(u, v)) \cap V_i$.

On the other hand, when $i = i_{u, v}$, this equality no longer holds because the right hand side is still a subset of the left hand side and $\ell_2(u, v)$ is in the left hand side but not in the right hand side. 
Therefore, by checking Equation~\eqref{eq:lca_reachability_2_listing} for $i = n/L, n/L -1, \dots$, we can find $i_{u, v}$. 

For each $i$, construct an $n \times L$ matrix $A$ whose rows are indexed by $V$ and columns are indexed by $V_i$ such that $A[u, x] = D[x, u]$, where $D$ is the transitive closure of $G$. We can then compute the matrix product $B = A A^T$ in $O(n^{\omega(1, \log_n L, 1)})$ time. Note that \[B[u, v] = \sum_{x \in V_i} A[u, x] \cdot A[v, x] = |\Anc(u) \cap \Anc(v) \cap V_i|.\]
Moreover, we can compute $|\Anc(x) \cap V_i|$ for all $x \in V$ in $O(n^2)$ time. Therefore, for each $i$, we can test if Equation~\eqref{eq:lca_reachability_2_listing} holds for all $u, v \in V$ in $O(n^2 + n^{\omega(1, \log_n L, 1)}) = O(n^{\omega(1, \log_n L, 1)})$ time.

Once we have computed $i_{u, v}$ for $u, v \in V$, we can simply search for the topologically latest vertex $x$ in $V_{i_{(u, v)}}$ such that $x \leadsto u$, $x \leadsto v$ but $x \not \leadsto \ell_1(u, v)$. This can be done in $O(L)$ time for each $u, v \in V$ with a linear scan of $V_{i_{u, v}}$. 

Thus, the overall runtime is $O(n^\omega + \frac{n}{L} \cdot n^{\omega(1, \log_n L, 1)} + n^2L)$. By setting $L=n^\lambda$ where $\lambda$ satisfies the equation $\omega(1, \lambda, 1) = 1 + 2\lambda$, we have that the runtime of the algorithm is in fact $O(n^\omega + n^{2 + \lambda})$, as desired. 

\subparagraph*{Algorithm for \apkarglca{3}.}
For $k = 3$, we first find $\ell_1(u, v)$ and $\ell_2(u, v)$ for all $u, v \in V$. Proceed exactly as we did for \apkarglca{2} to find the index $i_{u,v}$ such that $V_{i_{u,v}}$ contains $\ell_3(u, v)$, except we replace Equation~\eqref{eq:lca_reachability_2_listing} with the following equality:
\begin{equation}\label{eq:lca_reachability_3_listing}
    |\Anc(u) \cap \Anc(v) \cap V_i| = \left|\left(\Anc(\ell_1(u, v)) \cup \Anc(\ell_2(u, v))\right) \cap V_i\right|,
\end{equation}
which holds for $i > i_{u, v}$ and does not hold for $i=i_{u, v}$. We can also compute \[\left|\left(\Anc(\ell_1(u, v)) \cup \Anc(\ell_2(u, v))\right) \cap V_i\right|\] by multiplying an $n \times L$ matrix by an $L \times n$ matrix, thereby giving us the same overall runtime.
\end{proof}


\section{Lower Bounds}
\label{sec:lower_bound}

In this section, we show our conditional lower bounds for \eqlca{k} and \countlca{}. These lower bounds are the first conditional lower bounds for LCA problems that are higher than $n^{\omega-o(1)}$.

\subsection{Lower Bounds for \texorpdfstring{\eqlca{k}}{AP-Exactk-LCA}}

First, we show lower bounds for the \eqlca{k} problem by reducing from 3-uniform hypercliques. Combined with Corollary~\ref{cor:ktok+1}, the following theorem also shows that, for all constant $k \geq 6$, \eqlca{k} requires $n^{3-o(1)}$ time.

\begin{theorem}
\label{thm:apeq3456lb}
Assuming the \hyperclique{4}{3} hypothesis, \eqlca{3} requires $n^{2.5-o(1)}$ time. Assuming the \hyperclique{5}{3} hypothesis, \eqlca{4} and \eqlca{6} require $n^{8/3-o(1)}$ and $n^{3-o(1)}$ time respectively. Also, assuming the \hyperclique{6}{3} hypothesis, \eqlca{5} requires $n^{14/5-o(1)}$ time. 
\end{theorem}

\begin{proof}
All four reductions share the same underlying ideas, so we do not give full details for all of them. 

\subparagraph*{\hyperclique{4}{3} $\rightarrow$ \eqlca{3}.}
    Suppose we are given a $3$-uniform $4$-partite hypergraph $G$ on vertex sets $A, B, C, U$, where $|A| = |B| = |C| = \sqrt{n}$, and $|U| = n$. By Fact~\ref{fac:unbalanced_hyperclique}, the \hyperclique{4}{3} hypothesis implies that it requires $(|A||B||C||U|)^{1-o(1)} = n^{2.5-o(1)}$ time to determine whether $G$ contains a $4$-hyperclique. 
    
    We construct the following instance of \eqlca{3} as depicted in \autoref{fig:ap_eq3_lca_reduction}. The graph $G'$ contains $3$ layers of vertices $V_1, V_2, V_3$. Vertex set $V_1$ is a copy of $U$, vertex set $V_2$ equals $(A \times B) \sqcup (B \times C) \sqcup (C \times A)$ and vertex set $V_3$ equals $(A \times B) \sqcup C$. To distinguish vertices from $V_2$ and $V_3$, we use subscript $2$ and $3$, e.g. $(a, b)_2$ and $(a,b)_3$, to denote vertices from $V_2$ and $V_3$ respectively. 
    
    We also add the following edges to the graph $G'$:
    \begin{itemize}
        \item Add a directed edge from every vertex in $V_1$ to every vertex in $V_3$. 
        \item Add a directed edge from any vertex in $V_2$ to any vertex in $V_3$ as long as they do not have \textit{inconsistent} labels. For instance, for every $a \in A, b \in B, c \in C$, we add an edge from $(a, b)_2$ to $(a, b)_3$ and to $c_3$, but we do not add an edge from $(a, b)_2$ to $(a, b')_3$ if $b \ne b'$.
        \item For every $u \in V_1$ and every $(x, y)_2 \in V_2$, add a directed edge from $u$ to $(x, y)$ if and only if there is \textit{not} a $3$-hyperedge among $u, x$ and $y$. 
    \end{itemize}
    
    \begin{figure}[ht]
    \centering
    \scalebox{0.8}{
    \begin{tikzpicture}
    \def\W{3.5};
    \def\H{1.5};
    \def\colorone{violet};
    \def\colortwo{red};
	\def\colorthree{blue};
	\def\colorfour{gray};
	
    \node[draw,ellipse,minimum height=\H cm,minimum width=\W cm] (V1) at (0,0) {};
    \node[above right] at (V1.25) {$U$};
	\node at (0,0) [circle,fill, inner sep = 1pt, label=above right:$u$] (v1){};
	
     \node[draw,ellipse,minimum height=\H cm,minimum width=\W cm] (V21) at (-4,-3) {};
     \node[above left] at (V21.110) {$A \times B$};
     \node at (-4,-3) [circle,fill, inner sep = 1pt, label=left:{$(a', b')_2$}] (v21){};

	\node[draw,ellipse,minimum height=\H cm,minimum width=\W cm] (V22) at (0,-3) {};
    \node[above left] at (V22.110) {$B \times C$};
    \node at (0,-3) [circle,fill, inner sep = 1pt, label=right:{$(b'', c'')_2$}] (v22){};
    
    \node[draw,ellipse,minimum height=\H cm,minimum width=\W cm] (V23) at (4,-3) {};
    \node[above] at (V23.45) {$C \times A$};
    \node at (4,-3) [circle,fill, inner sep = 1pt, label=right:{$(c''', a''')_2$}] (v23){};

    \node[draw,ellipse,minimum height=\H cm,minimum width=\W cm] (V31) at (-2.5,-7) {};
    \node[below] at (V31.270) {$A \times B$};
    \node at (-2.5,-7) [circle,fill, inner sep = 1pt, label=below right:{$(a, b)_3$}] (v31){};

    \node[draw,ellipse,minimum height=\H cm,minimum width=2 cm] (V32) at (2.5,-7) {};
    \node[below] at (V32.270) {$C$};
    \node at (2.5,-7) [circle,fill, inner sep = 1pt, label=below right:{$c_3$}] (v32){};

	\draw[->,line width=1pt, color=\colorone] (v1) to[]  node[left] {\footnotesize	{$\{u, a', b'\} \not \in E$}} (v21);

	\draw[->,line width=1pt, color=\colorone] (v1) to[]  node[right, pos=0.6, xshift=-2] {\footnotesize	{$\{u, b'', c''\} \not \in E$}} (v22);

	 \draw[->,line width=1pt, color=\colorone] (v1) to[]  node[right, pos=0.4, xshift=1] {\footnotesize	{$\{u, c''', a'''\} \not \in E$}} (v23);

	\draw[->,line width=1pt, text width=0.9cm, color=\colortwo] (v21) to[]  node[left] {\footnotesize	{$a'=a$ $b'=b$}} (v31);
	
	\draw[->,line width=1pt, color=\colorthree] (v21) to[]  node[left, pos=0.4, xshift=-3, yshift=-1] {\footnotesize{all}} (v32);

	\draw[->,line width=1pt,color=\colortwo] (v22) to[]  node[left, pos=0.3, yshift=2] {\footnotesize	{$b''=b$}} (v31);
	\draw[->,line width=1pt, color=\colorthree] (v22) to[]  node[left, pos=0.3] {\footnotesize	{$c''=c$}} (v32);
		
	\draw[->,line width=1pt, color=\colortwo] (v23) to[]  node[right, pos=0.3] {\footnotesize	{$a'''=a$}} (v31);
	\draw[->,line width=1pt, color=\colorthree] (v23) to[]  node[right, pos = 0.4] {\footnotesize	{$c'''=c$}} (v32);
	
	\draw[->,line width=1pt, out=190, in=150, distance=7cm, color=\colorfour] (v1) to[]  node[left, pos = 0.4, xshift=-2, yshift=1] {\footnotesize{all}} (v31);
	\draw[->,line width=1pt, out=-10, in=30, distance=7cm, color=\colorfour] (v1) to[]  node[right, pos = 0.4, xshift=2, yshift=1] {\footnotesize{all}} (v32);

    \node[] () at (-8,0) {$V_1 $};
    \node[] () at (-8,-3) {$V_2$};
    \node[] () at (-8,-7) {$V_3$};
    
    \end{tikzpicture}}
    \caption{Construction of $G'$ in Theorem~\ref{thm:apeq3456lb} from the 3-uniform 4-hyperclique instance. Between the parts where we mark ``all'', we add all possible edges. Between the parts where we mark a condition, we only add an edge when the corresponding condition holds.} 
    \label{fig:ap_eq3_lca_reduction}
\end{figure}
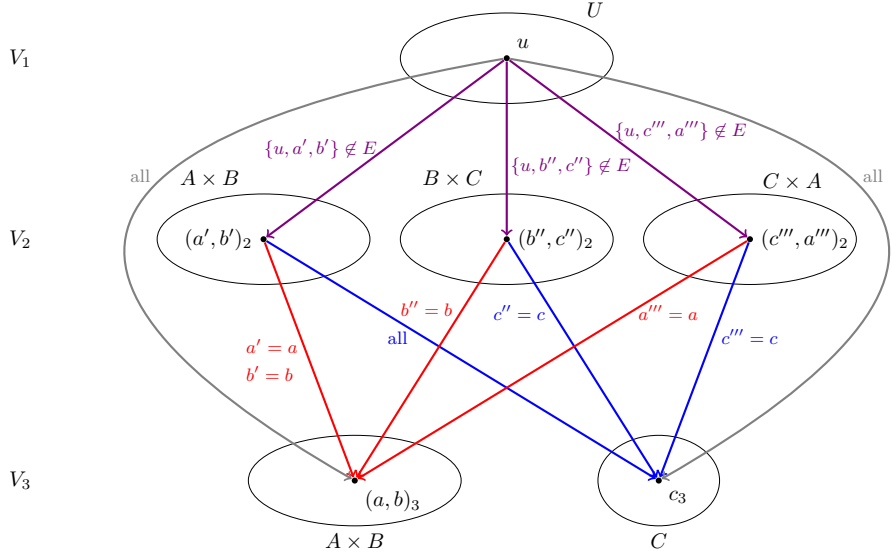

    We consider the set of LCAs for every pair of $(a, b)_3, c_3 \in V_3$.
    
    First, since $G'$ is a three layered graph, all common ancestors of $(a, b)_3$ and $c_3$ in $V_2$ are their LCAs. Since we only add edges from $V_2$ to $V_3$ when the labels are consistent, it is easy to verify that the set of LCAs in $V_2$ is $\{(a, b)_2, (b, c)_2, (c, a)_2\}$. 
    
    Now we claim that $a, b, c$ are in a $4$-hyperclique in $G$ if and only if 
    $(a, b)_3$ and $c_3$ have an LCA in $V_1$ 
    and there is a $3$-hyperedge among $a, b, c$ in $G$.
    
   Suppose $a, b, c$ are in a $4$-hyperclique with a vertex $u \in V(G)$. Since $G$ is $4$-partite, we must have $u \in U$. The copy of $u$ in $G'$ is clearly a common ancestor of $(a, b)_3$ and $c_3$, since we add all possible edges from $V_1$ to $V_3$. Because $a, b, c, u$ is in a $4$-hyperclique, $\{a, b, u\}, \{b, c, u\}, \{c, a, u\} \in E(G)$. Therefore, in $G'$ we do not add edges from $u$ to any of $(a, b)_2, (b,c)_2$ and $(c,a)_2$. Since these are the only common ancestors of $(a, b)_3$ and $c_3$ in $V_2$, $u$ in fact cannot reach any other vertex that can reach both $(a, b)_3$ and $c_3$, which makes $u$ an LCA. Clearly, there is a $3$-hyperedge among $a, b, c$ in $G$.

    To prove the converse, suppose $u \in V_1$ is an LCA of $(a, b)_3$ and $c_3$ and there is a $3$-hyperedge among $a, b, c$ in $G$. In that case, $u$ cannot reach any vertex that can reach both $(a, b)_3$ and $c_3$. In particular, $u$ cannot reach any of $(a, b)_2, (b, c)_2, (c, a)_2$. When we add edges from $V_1$ to $V_2$, we have that $\{a, b, u\}, \{b, c, u\}, \{c, a, u\}$ are all $3$-hyperedges in $G$. Also, since $\{a, b, c\}$ is a $3$-hyperedge,  there is indeed a $4$-hyperclique with vertices $a, b, c, u$.

    Thus,  $a, b, c$ are in a $4$-hyperclique in $G$ if and only if  the number of LCAs of 
    $(a, b)_3$ and $c_3$ is not $3$
    and there is a $3$-hyperedge among $a, b, c$ in $G$. Thus, given the result of an \eqlca{3} computation of $G'$, we can easily determine if $G$ has a  $4$-hyperclique.  Therefore, assuming the \hyperclique{4}{3} hypothesis, \eqlca{3} requires $n^{2.5-o(1)}$ time. 

\subparagraph*{\hyperclique{5}{3} $\rightarrow$ \eqlca{4}.}
Assuming the \hyperclique{5}{3} hypothesis,  \hyperclique{5}{3} on  graphs $G$ with $5$ parts $A, B, C, D, U$ such that $|A| = n^{2/3}, |B|=|C|=|D|=n^{1/3}$ and $|U| = n$ requires $n^{8/3-o(1)}$ time. Given such a hypergraph $G$, we construct the following  \eqlca{4} instance $G'$ on $O(n)$ vertices. 

The graph $G'$ contains $3$ layers of vertices $V_1, V_2, V_3$. Set $V_1$ is a copy of $U$, set $V_2$ equals $(A \times B) \sqcup (A \times C) \sqcup (A \times D) \sqcup (B \times C \times D)$ and set $V_3$ equals $(A \times B) \sqcup (C \times D)$. To distinguish vertices from $V_2$ and $V_3$, we use subscript $2$ and $3$ to denote vertices from $V_2$ and $V_3$ respectively. 

We also add the following edges to the graph $G'$:
\begin{itemize}
    \item Add a directed edge from every vertex in $V_1$ to every vertex in $V_3$. 
    \item Add a directed edge from some vertex in $V_2$ to some vertex in $V_3$ as long as they don't have inconsistent labels. 
    \item For every $u \in V_1$ and every $(a, x)_2 \in V_2$, add a directed edge from $u$ to $(a, x)_2$ if and only if there is \textit{not} a $3$-hyperedge among $u, a$ and $x$. For every $u \in V_1$ and every $(b, c, d)_2 \in V_2$,  add a directed edge from $u$ to $(b, c, d)_2$ if and only if there is \textit{not} a $4$-hyperclique among $\{u, b, c, d\}$. 
\end{itemize}

For any $a \in A, b \in B, c \in C, d \in D$, the set of LCAs in $V_2$ of $(a, b)_3$ and $(c,d)_3$ is clearly $\{(a, b)_2, (a, c)_2, (a, d)_2, (b, c, d)_2\}$. 
We claim that $a, b, c, d$ are  in a $5$-hyperclique in $G$ if and only if $(a, b)_3$ and $(c,d)_3$ have some LCA in $V_1$ and
 $a, b, c, d$ are in a $4$-hyperclique in $G$.

Suppose $\{a, b, c, d, u\}$ is a $5$-hyperclique in $G$. Clearly, $u$ is a common ancestor of $(a, b)_3$ and $(c,d)_3$. Also, because of the way we add edges between $V_1$ and $V_2$, $u$ cannot reach any one of $(a, b)_2, (a, c)_2, (a, d)_2, (b, c, d)_2$, which are the only vertices in $V_2$ that can reach both $(a, b)_3$ and $(b,c,d)_3$. Thus, $u$ is an LCA. 

For the converse, suppose $u \in V_1$ is an LCA and $a, b, c, d$ are in a $4$-hyperclique. Then $u$ must not reach any other vertex that can reach both  $(a, b)_3$ and $(c,d)_3$. In particular, $u$ cannot reach
 any of $(a, b)_2, (a, c)_2, (a, d)_2, (b, c, d)_2$. This means that $\{a, b, u\}, \{a, c, u\}, \{a, d, u\} \in E(G)$ and $\{b, c, d, u\}$ is in a $4$-hyperclique. These hyperedges and the hyperclique (together with the $4$-hyperclique among $a, b, c, d$) are enough to make $\{a, b, c, d, u\}$ a $5$-hyperclique. 

Thus, $G$ has a $5$-hyperclique if and only if there are some $a \in A, b \in B, c \in C, d \in D$ such that  $a, b, c, d$ are in a $4$-hyperclique and the number of LCAs of $(a, b)_3$ and $(b,c,d)_3$ is not $4$. This completes the reduction to \eqlca{4} and thus shows an $n^{8/3-o(1)}$ lower bound for it. 

\subparagraph*{\hyperclique{5}{3} $\rightarrow$ \eqlca{6}.}
Assuming the \hyperclique{5}{3} hypothesis, \hyperclique{5}{3} on graphs $G$ with $5$ parts $A, B, C, D, U$ such that $|A| = |B| = |C| = |D| = n^{1/2}$ and $|U| = n$ requires $n^{3-o(1)}$ time. Given such a hypergraph $G$, we construct an  \eqlca{6} instance $G'$ on $O(n)$ vertices. 

The graph $G'$ contains $3$ layers of vertices $V_1, V_2, V_3$. Set $V_1$ is a copy of $U$, set $V_2$ equals $(A \times B) \sqcup (A \times C) \sqcup (A \times D) \sqcup (B \times C) \sqcup (B \times D) \sqcup (C \times D)$ and set $V_3$ equals $(A \times B) \sqcup (C \times D)$. 

The set of edges we add and the remaining proofs are similar to previous reductions, so we omit them for conciseness. At the end, we can solve the \hyperclique{5}{3} instance $G$ by calling an \eqlca{6} algorithm on $G'$ with $O(n^2)$ additional work, and thus showing an $n^{3-o(1)}$ lower bound for \eqlca{6}.

\subparagraph*{\hyperclique{6}{3} $\rightarrow$ \eqlca{5}.}
Assuming the \hyperclique{6}{3} hypothesis,  \hyperclique{6}{3} on graphs $G$ with $6$ parts $A, B, C, D, E, U$ such that $|A| = |B| = n^{1/5}$, $|C| = |D| = n^{2/5}$, $|E| = n^{3/5}$ and $|U| = n$ requires $n^{14/5-o(1)}$ time. Given such a hypergraph $G$, we construct an  \eqlca{5} instance $G'$ on $O(n)$ vertices.

The graph $G'$ contains $3$ layers of vertices $V_1, V_2, V_3$. Set $V_1$ is a copy of $U$, set $V_2$ equals $(A \times C \times D) \sqcup (B \times C \times D) \sqcup (A \times B \times E) \sqcup (C \times E) \sqcup (D \times E)$ and set $V_3$ equals $(A \times B \times C) \sqcup (D \times E)$. 

The set of edges we add and the remaining proofs are similar to previous reductions, so we omit them for conciseness. The key guarantee for correctness is that all hyperedges in $G$ involving $U$ can be captured by some edge in $V_1 \times V_2$.  At the end, we can solve the \hyperclique{6}{3} instance $G$ by calling an \eqlca{5} algorithm on $G'$ with $O(n^2)$ additional work, and thus showing an $n^{14/5-o(1)}$ lower bound for \eqlca{6}.
\end{proof}

\begin{remark}
Note that in all our reductions to \eqlca{k} for $3 \le k \le 5$, we only need to output the results for $o(n^2)$ pairs of $u$ and $v$. For instance, in the reduction from \hyperclique{4}{3} to \eqlca{3}, we only need to output whether $(u, v)$ has exactly $3$ LCAs for $u \in A \times B$ and $v \in C$. The total number of such pairs is only $O(n^{1.5})$. This is the main reason why we do not get $n^{3-o(1)}$ conditional lower bounds for \eqlca{k} for $3 \le k \le 5$. On the other hand, in the reduction to \eqlca{6}, we do have $\Theta(n^2)$ queries.
\end{remark}

Williams \cite{williamsmax2sat} showed that \maxsat{} reduces to $3$-uniform hypercliques. Lincoln, Vassilevska Williams and Williams~\cite{lincoln2018tight} further generalized this reduction to a reduction from Constraint Satisfaction Problem (CSP) on degree-$3$ formulas to $3$-uniform hypercliques. Therefore, Theorem~\ref{thm:apeq3456lb} also works assuming the \maxsat{} hypothesis or the hardness of maximizing the number of satisfying clauses in degree-$3$ CSP formulas. 

\begin{corollary}
Assuming \maxsat{} (or even max degree $3$ CSP formulas) on $N$ variables and $\poly(n)$ clauses requires $2^{N-o(N)}$ time, \eqlca{3}, \eqlca{4}, \eqlca{5} and \eqlca{6} requires $n^{2.5-o(1)}$, $n^{8/3-o(1)}$, $n^{14/5-o(1)}$ and $n^{3-o(1)}$ time respectively.  
\end{corollary}

\subsection{Lower Bounds for Counting LCAs}

In this section, we show two conditional lower bounds for \countlca{}, one based on SETH and one based on the \fourclique{} hypothesis. 

The next lemma is a crucial tool for the SETH lower bound. It is a generalization of our previous reduction from \hyperclique{5}{3}  to \eqlca{6}.

\begin{lemma}
\label{lem:SETH_lem}
If there exists a $T(N)$ time algorithm for \eqlca{\binom{2(k-1)}{k-1}} for graphs with $N$ vertices, then there exists an $O(f(k) \poly(n)^{f(k)} T(2^{n/3}))$ time algorithm for \maxksat{k} with $n$ variables and $\poly(n)$ clauses for some function $f$.
\end{lemma}

To prove the lemma, we first reduce \maxksat{k} to $k$-uniform  $(2k-1)$-hyperclique, which is a straightforward generalization of Williams' \maxksat{2} algorithm \cite{williamsmax2sat}. Then we reduce $k$-uniform  $(2k-1)$-hyperclique to \eqlca{\binom{2(k-1)}{k-1}}, building on ideas similar to the proof of Theorem~\ref{thm:apeq3456lb}. 

\begin{proof}
We first reduce \maxksat{k} to $k$-uniform  $(2k-1)$-hyperclique. This is a straightforward generalization of Williams' \maxksat{2} algorithm \cite{williamsmax2sat}, but we include the proof for completeness. 

Following Williams' reduction, we first reduce \maxksat{k} to max-weight  \hyperclique{2k-1}{k}.
We split the set of variables to $(2k-1)$ parts, $A_0 \sqcup A_1 \sqcup A_2 \sqcup \cdots \sqcup A_{2k-2}$, where $|A_{0}| = \frac{n}{3}$ and $|A_i| = \frac{n}{3(k-1)}$ for each $i \in [2k-2]$. We create a $(2k-1)$-partite $k$-uniform hypergraph on vertices $P_0 \sqcup \cdots \sqcup P_{2k-2}$, where each part corresponds to a set of variables. The vertices in part $P_i$ corresponds to $A_i$, representing partial assignments $\alpha_i$  of variables in $A_i$. Initially, all hyperedge weights are $0$. 

For each clause of the \maxksat{k} instance, we can identify $k$ subsets of variables $A_{i_1}, \ldots, A_{i_k}$ that contain all the variables in the clause (if there are multiple possible ways to choose these $k$ subsets, we can just pick any one of them). For every possible partial assignment $\alpha_{i_1}, \ldots, \alpha_{i_k}$, we increment the weight of its hyperedge if the partial assignment satisfy the clause. It is then easy to see that the weight of a hyperclique $(\alpha_{0}, \ldots, \alpha_{2k-2})$ equals the number of clauses the assignment satisfies. Thus, the weight of the max-weight hyperclique equals the maximum number of clauses that can be satisfied. The time to construct this hypergraph is $O(\poly(n) \cdot 2^{\frac{2}{3} n})$. 

Since the hyperedge weights are bounded by $n^{O(1)}$, we can then reduce this max-weight  \hyperclique{2k-1}{k} instance to $n^{O\left(\binom{2k-1}{k}\right)}$ instances of  \hyperclique{2k-1}{k} by enumerating all the combinations of edge weights. 

Finally, we reduce each  \hyperclique{2k-1}{k} instance on vertex set $P_0 \sqcup P_1 \sqcup \cdots \sqcup P_{2k-3} \sqcup P_{2k-2}$ to \eqlca{\binom{2(k-1)}{k-1}}. 

Similar to the proof of Theorem~\ref{thm:apeq3456lb}, 
we create a  graph $G'$ with $3$ layers of vertices $V_1, V_2, V_3$. We set $V_1$ to be a copy of $P_{0}$. We set $V_2$ to be $\bigsqcup_{S \in \binom{[2k-2]}{k-1}} \prod_{i \in S} P_i$, i.e., vertices in $V_2$ are tuples of $(k-1)$ vertices from $k-1$ distinct vertex sets of $P_1, \ldots, P_{2k-1}$. Finally, we set $V_3$ to be $(\prod_{1 \le i \le k-1} P_i) \sqcup (\prod_{k \le i \le 2k-2} P_i)$, i.e., vertices in $V_3$ are either $(\alpha_1, \ldots, \alpha_{k-1}) \in P_1 \times \cdots \times P_{k-1}$ or $(\alpha_k, \ldots, \alpha_{2k-2}) \in P_{k} \times \cdots \times P_{2k-2}$.

We also add the following edges to the graph $G'$:
\begin{itemize}
    \item Add a directed edge from every vertex in $V_1$ to every vertex in $V_3$. 
    \item Add a directed edge from some vertex in $V_2$ to some vertex in $V_3$ as long as they don't have inconsistent labels. 
    \item From every $\alpha_{0} \in V_1$ to every $(\alpha_{i_1}, \ldots, \alpha_{i_{k-1}}) \in V_2$,  add a directed edge  if and only if there is \textit{not} a $k$-hyperedge among these  vertices. 
\end{itemize}

Similar to the proof of Theorem~\ref{thm:apeq3456lb}, we can show that $\alpha_0, \ldots, \alpha_{2k-2}$ is in a $(2k-1)$-hyperclique if and only if $\alpha_1, \ldots, \alpha_{2k-2}$ form a $(2k-2)$-hyperclique and the number of LCAs of $(\alpha_1, \ldots, \alpha_{k-1}) \in V_3$ and $(\alpha_{k}, \ldots, \alpha_{2k-2}) \in V_3$ is not $\binom{2k-2}{k-1}$. We omit the details of the proof for conciseness. 

The number of vertices in the \eqlca{\binom{2(k-1)}{k-1}} instance is $O(\binom{2k-2}{k-1} \cdot 2^{n/3})$. Thus, if there is a $T(N)$ time algorithm for \eqlca{\binom{2(k-1)}{k-1}} on a graph with $N$ vertices, then the overall running time of the algorithm for \maxksat{k} is 
\[O\left(\poly(n) \cdot 2^{\frac{2}{3} n} + n^{O\left(\binom{2k-1}{k}\right)} \cdot T\left(\binom{2k-2}{k-1} \cdot 2^{n/3}\right)\right) = O(f(k) \poly(n)^{f(k)} T(2^{n/3}))\]
for some function $f$, 
as claimed.
\end{proof}

\begin{remark}
Lemma~\ref{lem:SETH_lem} implies that if we assume the \maxksat{k} hypothesis, then \eqlca{\binom{2(k-1)}{k-1}} requires $n^{3-o(1)}$ time. Since our reduction uses  \hyperclique{2k-1}{k} as an intermediate problem, the same lower bound also holds assuming the \hyperclique{2k-1}{k} hypothesis. 
\end{remark}

Now we show our SETH lower bound using Lemma~\ref{lem:SETH_lem}.

\SETH*

\begin{proof}
For the sake of contradiction, assume \countlca\ has an $O(n^{3-\epsilon})$ time algorithm for $\epsilon > 0$ when the algorithm only needs to return the minimum between the count and $g(n)$. For any fixed $k$, when $n$ is large enough, we have $\binom{2(k-1)}{k-1} < g(n)$, so we can solve \eqlca{\binom{2(k-1)}{k-1}} in $O(n^{3-\epsilon})$ time. 
Thus, 
by Lemma~\ref{lem:SETH_lem}, we can solve \maxksat{k} (and thus $k$-SAT) with $n$ variables and $\poly(n)$ clauses in time 
\[O(f(k) \poly(n)^{f(k)} (2^{n/3})^{3-\epsilon}) = O(f(k) \poly(n)^{f(k)} 2^{(1-\epsilon/3) n}) = O(\poly(n) \cdot 2^{(1 - \epsilon/3)n}),\]
which would refute SETH. 
\end{proof}

Finally, we present our reduction from \fourclique{} to \countlca{}, showing an $n^{\omega(1, 2, 1) - o(1)}$ lower bound for \countlca{} assuming the current algorithm for \fourclique{} is optimal. 

\FourCliqueHardness*

\begin{proof}
    
    Suppose we are given a \fourclique{} instance $G = (V, E)$. Without loss of generality, we assume $G$ is a $4$-partite graph with four vertex parts $V = A \sqcup B \sqcup C \sqcup D$ of size $n$ each.
    
    First, make a copy  $G' = (V', E')$ of $G$, and modify the edge set of $G'$ as follows:
    \begin{itemize}
        \item Remove all edges between $A$ and $B$.
        \item Direct all edges from $D$ to $A$ and $B$.
        \item Direct all edges from $C$ to $A, B$ and $D$.
    \end{itemize}
    
    Then we add two additional vertex sets $A'$ and $B'$ to $G'$, where $A'$ is a copy of $A$ and $B'$ is a copy of $B$. We use $a'$ to denote the copy of $a \in A$ in $A'$ and use $b'$ to denote the copy of $b \in B$ in $B'$. We also add the following edges:
    \begin{itemize}
        \item For every $a \in A$, add an edge $(a', a)$. 
        \item For every $b \in B$, add an edge $(b', b)$.  
        \item For every $a \in A, b \in B$, add two edges $(a', b)$ and $(b', a)$. 
        \item For every $a \in A, c \in C$, add an edge $(c, a')$ if $\{c, a\} \not \in E$.
        \item For every $b \in B, c \in C$, add an edge $(c, b')$ if $\{c, b\} \not \in E$.
    \end{itemize}
        
    This construction of the graph is also depicted in Figure~\ref{fig:4_clique_reduction}. From there, it is clear that $G'$ is a $3$-layered graph. 
    
    \begin{figure}[ht]
    \centering
    \scalebox{0.8}{
    \begin{tikzpicture}
    \def\W{3.5};
    \def\H{1.5};
    \def\colorone{violet};
    \def\colortwo{red};
	\def\colorthree{blue};
	\def\colorfour{gray};
	
    \node[draw,ellipse,minimum height=\H cm,minimum width=\W cm] (V1) at (0,0) {};
    \node[above right] at (V1.25) {$C$};
	\node at (0,0) [circle,fill, inner sep = 1pt, label=above right:$c$] (v1){};
	
     \node[draw,ellipse,minimum height=\H cm,minimum width=\W cm] (V21) at (-4,-3) {};
     \node[above left] at (V21.110) {$A'$};
     \node at (-4,-3) [circle,fill, inner sep = 1pt, label=left:{$a'$}] (v21){};

	\node[draw,ellipse,minimum height=\H cm,minimum width=\W cm] (V22) at (0,-3) {};
    \node[above left] at (V22.110) {$D$};
    \node at (0,-3) [circle,fill, inner sep = 1pt, label=right:{$d$}] (v22){};
    
    \node[draw,ellipse,minimum height=\H cm,minimum width=\W cm] (V23) at (4,-3) {};
    \node[above] at (V23.45) {$B'$};
    \node at (4,-3) [circle,fill, inner sep = 1pt, label=right:{$b'$}] (v23){};

    \node[draw,ellipse,minimum height=\H cm,minimum width=\W cm] (V31) at (-2.5,-7) {};
    \node[below] at (V31.270) {$A$};
    \node at (-2.5,-7) [circle,fill, inner sep = 1pt, label=below right:{$a$}] (v31){};

    \node[draw,ellipse,minimum height=\H cm,minimum width=\W cm] (V32) at (2.5,-7) {};
    \node[below] at (V32.270) {$B$};
    \node at (2.5,-7) [circle,fill, inner sep = 1pt, label=below right:{$b$}] (v32){};

	\draw[->,line width=1pt, color=\colorone] (v1) to[]  node[left] {\footnotesize	{$\{c, a\} \not \in E$}} (v21);

	\draw[->,line width=1pt, color=\colorfour] (v1) to[]  node[right, pos=0.6, xshift=-2] {\footnotesize	{$\{c, d\} \in E$}} (v22);

	 \draw[->,line width=1pt, color=\colorone] (v1) to[]  node[right, pos=0.4, xshift=1] {\footnotesize	{$\{c, b\} \notin E$}} (v23);

	\draw[->,line width=1pt, text width=0.9cm, color=\colortwo] (v21) to[]  node[left] {\footnotesize{$a' = a$}} (v31);
	
	\draw[->,line width=1pt, color=\colorthree] (v21) to[]  node[left, pos=0.4, xshift=-3, yshift=-1] {\footnotesize{all}} (v32);

	\draw[->,line width=1pt,color=\colorfour] (v22) to[]  node[left, pos=0.3, yshift=6, xshift=6] {\footnotesize	{$\{d, a\} \in E$}} (v31);
	\draw[->,line width=1pt, color=\colorfour] (v22) to[]  node[right, pos=0.3,yshift=6,xshift=-3] {\footnotesize	{$\{d, b\} \in E$}} (v32);
		
	\draw[->,line width=1pt, color=\colorthree] (v23) to[]  node[right, pos=0.3,yshift=-5] {\footnotesize{all}} (v31);
	\draw[->,line width=1pt, color=\colortwo] (v23) to[]  node[right, pos = 0.4] {\footnotesize	{$b' = b$}} (v32);
	
	\draw[->,line width=1pt, out=190, in=150, distance=7cm, color=\colorfour] (v1) to[]  node[left, pos = 0.4, xshift=-2, yshift=1] {\footnotesize{$\{c, a\} \in E$}} (v31);
	\draw[->,line width=1pt, out=-10, in=30, distance=7cm, color=\colorfour] (v1) to[]  node[right, pos = 0.4, xshift=2, yshift=1] {\footnotesize{$\{c, b\} \in E$}} (v32);

    \end{tikzpicture}}
    \caption{Construction of $G'$ in Theorem~\ref{thm:4_clique_hardness} given a 4-partite \fourclique{} instance. Between parts where we mark ``all'', we add all possible edges. Between parts where we mark a condition, we only add an edge when the corresponding condition holds.} 
    \label{fig:4_clique_reduction}
\end{figure}

    \begin{claim}
    \label{cl:count_lca_4_clique}
    For every $a \in A, b \in B, c \in C$, $c$ is an LCA of $a$ and $b$ in $G'$ if and only if $\{c, a\}, \{c, b\} \in E$ and there doesn't exist any $d \in D$ such that $\{c, d\}, \{d, a\}, \{d, b\} \in E$. 
    \end{claim}
    \begin{proof}
    First, suppose $c$ is an LCA of $a$ and $b$. For the sake of contradiction, suppose $\{c, a\} \not \in E$. Then by the construction of $G'$, $(c, a') \in E'$. Also, $(a', a), (a', b) \in E'$, so $c$ cannot be an LCA. This leads to a contradiction, so we must have $\{c, a\} \in E$. Similarly, we must have $\{c, b\} \in E$. Finally, suppose for the sake of contradiction that there exists a $d \in D$ such that $\{c, d\}, \{d, a\}, \{d, b\} \in E$, then by construction, $(c, d), (d, a), (d, b) \in E'$, so $c$ cannot be an LCA. Thus, there doesn't exist any $d \in D$ such that $\{c, d\}, \{d, a\}, \{d, b\} \in E$.
    
    Now we prove the converse direction. Suppose $\{c, a\}, \{c, b\} \in E$ and there doesn't exist any $d \in D$ such that $\{c, d\}, \{d, a\}, \{d, b\} \in E$. By our construction, $(c, a), (c, b) \in E'$, so $c$ is at least a common ancestor of $a$ and $b$. Since $G'$ is a $3$-layered graph, it suffices to show that there isn't any vertex $u$ in the middle layer such that $(c, u), (u, a), (u, b) \in E'$. First, for any $u \in A'$, if $u \ne a'$, then $(u, a) \not \in E'$; if $u = a'$, then $(c, u) \not \in E'$ because $\{c, a\} \in E$. Therefore, there isn't any $u \in A'$ such that $(c, u), (u, a), (u, b) \in E'$. Similarly, there isn't any $u \in B'$ such that $(c, u), (u, a), (u, b) \in E'$. For any $d \in D$, we already have the condition that at least one of $\{c, d\}, \{d, a\}, \{d, b\}$ is not in $E$, so at least one of $(c, d), (d, a), (d, b)$ is not in $E'$. Therefore, $c$ is an LCA. 
    \end{proof}
    
    Using this claim, we describe our algorithm below. 
    
    First, run \countlca{} to compute $|\lca(a, b)|$ for all $(a, b) \in A \times B$. 
    Since $G'$ is a three-layered graph, the set of LCAs of $a$ and $b$ in the middle layer is exactly the set of their common neighbors in the middle layer. Therefore, we can easily compute $|\lca(a, b) \cap (A' \cup B' \cup D)|$ in $O(n^\omega)$ time by using matrix multiplication to count the number of their common neighbors in the middle layer. Also, clearly, there isn't any LCA of $a$ and $b$ in $A$ or $B$. Thus, we can compute the number of $c \in C$ that is an LCA of $a$ and $b$ by 
    \[|\lca(a, b) \cap C| = |\lca(a, b)| - |\lca(a, b) \cap (A' \cup B' \cup D)|.\]
    
    By Claim~\ref{cl:count_lca_4_clique}, $|\lca(a, b) \cap C|$ is exactly the number of $c \in C$ such that $\{c, a\}, \{c, b\} \in E$ and there doesn't exist any $d \in D$ such that $\{c, d\}, \{d, a\}, \{d, b\} \in E$. 
    
    Next, in $O(n^\omega)$ we can use matrix multiplication again to compute $Q(a, b)$ for every $(a, b)$ where $Q(a, b)$ is defined as the number of $c \in C$ such that $\{c, a\}, \{c, b\} \in E$. 

    Note that $Q(a, b) - |\lca(a, b) \cap C|$ is exactly the number of $c \in C$ such that $\{c, a\}, \{c, b\} \in E$ and there \textit{exists} $d \in D$ such that $\{c, d\}, \{d, a\}, \{d, b\} \in E$. Thus, $a$  and $b$ are in a 4-clique if and only if $\{a, b\} \in E$ and $Q(a, b) - |\lca(a, b) \cap C| > 0$.

    Overall, if we can compute \countlca{} in $T(n)$ time, then we can solve the $4$-partite \fourclique{} instance in $O(T(n) + n^\omega)$ time.
\end{proof}

Since \countlca{} is easier than \alllca{}, this lower bound shows that the \alllca{} algorithm in \cite{Eckhardt2007FastLC} is in fact conditionally optimal.

\section{\apvlca}
\label{sec:apvlca}

In this section, we show two conditional lower bounds for \apvlca{}, based on the \maxwitness{} hypothesis and the \hyperclique{4}{3} hypothesis. First, recall the following theorem:

\apvlbmaxwitness*

At the high level, we reduce the $\textsf{MaxWitness}$ problem to $O(\log n)$ calls of the \apvlca{} problem using a parallel binary search technique. 

\begin{proof}
    Without loss of generality, suppose $n = 2^\ell$ for some integer $\ell$. Suppose we have two $n \times n$ Boolean matrices $A$ and $B$, each already padded with a column and row of ones to ensure that there always exists a Boolean witness. Now, we will describe an algorithm to compute $C =$ \maxwitness{}$(A, B)$ using an \apvlca{} algorithm $\ell = \log n$ times. At the high level, we will be using a parallel binary search to find the maximum witness corresponding to each entry of $C$. 
     
    Construct a tripartite graph $G$ on vertices $V = I \sqcup J \sqcup K$, where $|I| = |J| = |K|$ and identify each of the sets with $[n]$. Add a directed edge from $k \in K$ to $i \in I$ if $A[i, k] = 1$ and an edge from $k \in K$ to $j \in J$ if $B[k, j] = 1$. Then, computing $C[i, j]$ is the same as determining the largest $k \in K$ that is a common ancestor of both $i \in I$ and $j \in J$. 
    Now, we will iteratively find the $t$th bit in the binary representation of each $C[i, j]$ for $t = 1, \ldots, \ell$ (the first bit is the highest order bit, and the last bit is the lowest order bit). In the first iteration, we do the following.
    
    \textbf{Phase 1:} Construct a graph $G_1$ by first making a copy of $G$ and adding a vertex $w$. Then, we add a directed edge from $w$ to every vertex in $I \cup J$. Now, add a directed edge from $w$ to all vertices $k \in K$ whose binary representation starts with $1$. Finally, run \apvlca{} where we guess $w$ is an LCA for all pairs $(i, j) \in I \times J$. If $w$ is in fact an LCA, set $c^{(1)}_{i, j} = 0$, and otherwise, set $c^{(1)}_{i, j} = 1$.

    More generally, at the $t$th iteration of the algorithm, we do the following.
    
    \textbf{Phase $t$:} At the $t$th iteration of the algorithm for $1 \leq t \leq \ell$, construct the graph $G_t$ as follows. First, make a copy of $G$.
    Then, for each string $b = {b_1b_2 \dots b_{t-1}} \in \{0, 1\}^{t-1}$, create a vertex $w_{{b}}$. Now, add an edge from $w_{{b}}$ to all vertices in $K$ whose binary representation starts with ${b_1b_2 \dots b_{t-1}||1}$. Then, add an edge from every $w_{{b}}$ to every vertex in $I \cup J$. If $c^{(t-1)}_{i, j} = {b}$, guess that $w_{{b}}$ is an LCA for $(i, j) \in I \times J$. Run \apvlca{} with all of these guesses. If the algorithm outputs {\sc yes} for $(i, j)$, set $c^{(t)}_{i, j} = {{b} || 0}$. Otherwise, set $c^{(t)}_{i, j} = {{b} || 1}$.
    
    We show by induction that at Phase $t$, $c_{i, j}^{(t)}$ is the first $t$ bits of $C[i, j]$. In Phase 1, note that $w$ is an LCA for $(i , j) \in I \times J$ exactly when none of its children are common ancestors of $(i, j)$. In other words, $(i, j)$ has no common ancestor (and hence no witness) $k \in K$ whose first bit is 1.
    
    Suppose at iteration $t-1$, this claim is true. In other words, for each $i, j$, $d_{i, j} = c^{(t-1)}_{i, j}$ corresponds to the first $t-1$ bits of $C[i, j]$. Then, at iteration $t$, we guessed that $w_{d_{i, j}}$ is an ancestor. Since $w_{d_{i,j}}$ only has children whose first $t$ bits are $d_{i, j} || 1$, it is an LCA of $(i, j)$ exactly when none of these children are common ancestors, i.e. the largest common ancestor of $(i, j)$ has binary representation starting with $d_{i, j}||0$. Otherwise, it starts with $d_{i, j}||1$, as desired.
    
    Therefore, after $\ell$ iterations, we have that $C[i, j] = c^{(\ell)}_{i, j}$ (where we interpret $c_{i, j}$ as an $\ell$-bit binary integer). The algorithm does $O(n^2)$ work at each phase to construct $G_t$, and then invokes an \apvlca{} algorithm. Hence the overall runtime is $\tilde{O}(n^2 + T(n)) = \tilde{O}(T(n))$, as desired.
\end{proof}

\apvlbhyperclique*

\begin{proof}
Suppose we are given a $3$-uniform $4$-partite hypergraph $G$ on vertex sets $A, B, C, U$, where $|A| = |B| = |C| = \sqrt{n}$, and $|U| = n$. The \hyperclique{4}{3} hypothesis implies that it requires $n^{2.5-o(1)}$ time to determine whether $G$ contains a $4$-hyperclique by Fact~\ref{fac:unbalanced_hyperclique}.

We construct the following \vlca{} instance $G'$ on $O(n)$ vertices. 

The graph $G'$ contains $3$ layers of vertices $V_1, V_2, V_3$ with an additional vertex $s$. We set $V_1$ to be $A \times B$, set $V_2$ to be a copy of $U$ and set $V_3$ to be $(B \times C) \sqcup (C \times A)$. 

We also add the following edges to the graph $G'$.
\begin{itemize}
    \item Add a directed edge from every $v_1 \in V_1$ to every $v_3 \in V_3$. 
    \item Add a directed edge from $(a, b) \in V_1$ to $u \in V_2$ if and only if $\{u, a, b\} \in E(G)$. 
    \item Add a directed edge from $u \in V_2$ to $(b, c) \in V_3$ if and only if $\{u, b, c\} \in E(G)$. Similarly, add a directed edge from $u \in V_2$  to $(c, a) \in V_3$ if and only if $\{u, c, a\} \in E(G)$. 
    \item Add a directed edge from $s$ to every other vertex in $G'$. This ensures that every pair of vertices has some common ancestors, and thus has at least one LCA.  
\end{itemize}

We claim that for every $a \in A, b \in B, c \in C$, $a, b, c$ are in a $4$-hyperclique in $G$ if and only if $\{a, b, c\} \in E(G)$ and $(a, b)$ is \textit{not} an LCA of $(b, c)$ and $(c, a)$ in $G$. 

First, if $a, b, c$ are in a $4$-hyperclique with $u$, then clearly $\{a, b, c\} \in E(G)$. Also, by the construction of $G'$, $((a, b), u), (u, (b, c)), (u, (c, a))$ are all edges in $G'$. Thus, $(a, b)$ can reach a vertex $u$ which can reach both $(b, c)$ and $(c, a)$, so $(a, b)$ is not an LCA of $(b, c)$ and $(c, a)$. 

Conversely, if $\{a, b, c\} \in E(G)$ and $(a, b)$ is \textit{not} an LCA of $(b, c)$ and $(c, a)$, then since $(a, b)$ can reach both $(b, c)$ and $(c, a)$ via edges added from $V_1$ to $V_3$, $(a, b)$ must be able to reach some vertex that can reach both $(b, c)$ and $(c, a)$. Such a vertex must belong to $V_2$. Say the vertex is $u$, then 
 by the construction of $G'$, we must have $\{a, b, u\}, \{b, c, u\}, \{c, a, u\} \in E(G)$. Together with the hyperedge $\{a, b, c\}$, $a, b, c$  is in a $4$-hyperclique. 
 
 Therefore, we can run \vlca{} on $G'$ with the following set of LCA candidates:
 \begin{itemize}
     \item For every $a \in A, b \in B, c \in C$ such that $\{a, b, c\} \in E(G)$, let $w_{(b, c), (c, a)} = (a, b)$. 
     \item For every other pair of vertices $u, v \in V(G')$, we use Grandoni et al.'s algorithm \cite{grandoni2020lca} to find an actual LCA $\ell_{u,v}$ for them in $O(n^{2.447})$ time and set $w_{u, v} = \ell_{u, v}$. 
 \end{itemize}
 If some LCA candidate is incorrect, it must be that $(a, b)$ is not an LCA for $(b, c)$ and $(c, a)$ for some $a \in A, b \in B, c \in C$ such that $\{a, b, c\} \in E(G)$ and thus by previous discussion, the hypergraph $G$ has a $4$-hyperclique. On the other hand, if all LCA candidates are correct, then the hypergraph $G$ does not have a $4$-hyperclique.
 
 Therefore, assuming the \hyperclique{4}{3} hypothesis, \vlca{} requires $n^{2.5-o(1)}$ time.

\end{proof}

Our conditional lower bounds for \apvlca{} and \vlca{} are surprising because they suggest that \apvlca{} and \vlca{} require $n^{2.5-o(1)}$ time, while \aplca{} can be computed in $O(n^{2.447})$ time~\cite{grandoni2020lca}. This defies the common intuition that verification should be easier than computation.

\section{Open problems}
\label{sec:open}
We conclude this work by pointing out some potential future directions.
\begin{enumerate}
    \item Does there exist a subcubic time algorithm for \eqlca{k} for any $3 \leq k \leq 5$? Or, can we show an $n^{3-o(1)}$ conditional lower bound for \eqlca{k} for any such $k$? How about \apvlca{}? 
    \item Is it possible to show conditional lower bounds for \apkarglca{k} without using \geqlca{k} as an intermediate problem? For instance, since \geqlca{k} has $O(n^\omega)$ time algorithms for $k \le 3$, we cannot hope to get a higher than $n^\omega$ lower bound for \apkarglca{k} for $k \le 3$ using \geqlca{k} as an intermediate problem. However, the current best algorithm for \apkarglca{1} runs in $O(n^{2.447})$ and the best algorithm for \apkarglca{2} and \apkarglca{3} runs in $O(n^{2.529})$ time. 
    \item All our reductions reduce to instances of LCA variants in graphs with $O(1)$ layers. In such graphs, some variants could have faster algorithms. In particular, 
    \aplca{} has an $\tO(n^\omega)$ time algorithm \cite{CzumajKL07} for graphs with $O(1)$ layers, and thus we cannot hope to show a higher conditional lower bound using our techniques. In order to overcome this, we need to find reductions that show hardness for LCA variants in graphs with many layers.
    
    \item Are there any other related problems whose verification version is easier than the computation version? Can we reduce these problems to or from \aplca{}?
\end{enumerate}

\bibliography{ref}

\end{document}